\documentclass[12pt]{article}
\usepackage
[
    a4paper,
    margin = 2.5cm,
]
{geometry}
\usepackage{amsfonts}
\usepackage{amsmath}
\usepackage{amssymb}
\usepackage{amsthm}
\usepackage{graphicx}
\usepackage{tikz}
\usetikzlibrary{decorations.pathreplacing}
\usepackage{color}
\usepackage{mathpazo}
\usepackage{times}
\usepackage{caption}
\usepackage{subcaption}
\usepackage{natbib}
\usepackage{url}
\usepackage[multiple]{footmisc}

\theoremstyle{definition}
\newtheorem{assumption}{Assumption}

\newtheorem{definition}{Definition}
\newtheorem{example}{Example}
\newtheorem{lemma}{Lemma}
\newtheorem{proposition}{Proposition}

\DeclareMathOperator*{\argmax}{arg\,max}

\newcommand{\EE}{\mathbb{E}}

\DeclareMathOperator\supp{supp}

\newcommand{\prob}{\mathrm{Pr}}

\usepackage{hyperref}
\begin{document}

\title{Information Intermediaries in Monopolistic Screening \thanks{~We owe special thanks to Alessandro Pavan and Piotr Dworczak for their support and guidance throughout this project. We also wish to thank Zeinab Aboutalebi, Ian Ball, Sandeep Baliga, {\"O}zlem Bedre-Defolie, Georgy Egorov, Daniel Garcia, Maarten Janssen, Peter Klibanoff, Laurent Mathevet, Niko Matouschek, Jeffrey Mensch, Filippo Mezzanotti, Wojciech Olszewski, Harry Pei, Alvaro Sandroni, Karl Schlag, Marciano Siniscalchi, as well as various seminar audiences at Northwestern University, University of Vienna and European University Institute for valuable comments and suggestions. Kyriazis gratefully acknowledges the support received as part of the project, Digital Platforms: Pricing, Variety, and Quality Provision (DIPVAR), that has received funding from the European Research Council (ERC) under the European Union’s Horizon 2020 research and innovation programme (grant agreement No 853123) and from the Deutsche Forschungsgemeinschaft (DFG, German Research Foundation) under project number 504870592.}}

\author{Panagiotis Kyriazis\thanks{~Department of Economics, European University Institute (EUI). Email: panagiotis.kyriazis@eui.eu.}
\and Edmund Y. Lou\thanks{~Department of Accounting and Finance, University of Auckland. Email: edmund.lou@auckland.ac.nz.}}

\date{March 2026}

\maketitle
\begin{center}
\end{center}

\begin{abstract}
We investigate the relationship between product offerings, information dissemination, and consumer decision-making in a monopolistic screening environment in which consumers lack information about their valuation of quality-differentiated products. An intermediary, who is driven by the objective of maximizing consumer surplus but is also biased towards high-quality products, provides recommendations after the monopolist announces the menu of product choices. We characterize the monopolist's profit-maximizing finite-item menu. Our results show that as intermediaries place greater emphasis on consumer surplus over product quality, sellers are prompted to strategically expand their product range. Intriguingly, this augmented product variety decreases economic efficiency compared to scenarios where direct seller-to-consumer information provision is the norm. The role of information intermediaries proves pivotal in shaping consumer welfare, market profitability, and overarching economic efficiency. Our insights underscore the complexities introduced by these intermediaries that policymakers and market designers must consider when designing policies centered on consumer learning and market information transparency.
\end{abstract}
\newpage

\section{Introduction}
Consumer welfare is significantly influenced by the prevailing uncertainty that complicates consumers' selection among diverse products. This is particularly evident in the contemporary digital economy, where abundant information exists, capable of guiding consumer decisions and enhancing the matching between consumers and products. This is often facilitated through intermediaries or platforms showcasing these products. Consequently, the extent of information an intermediary opts to reveal plays a pivotal role in shaping market outcomes. This disclosure not only steers consumer choices but also influences sellers' profitability and, in the grand scheme, dictates the economic performance of a market. By exercising discretion over information transparency, intermediaries wield power over market dynamics and the resultant stakeholder benefits.

In this paper, we investigate the interplay between product offerings, information dissemination, and consumer decision-making in a classic screening setting, with a substantial difference: consumers, instead of being endowed with private information about their valuation of products, rely on an intermediary to learn about their preferences. The intermediary, while aiming to optimize the consumer's payoff, is also inclined toward steering the consumers to higher-quality products. Aware of the consumer's dependency on the intermediary for information and the learning that ensues once products are unveiled, the seller strategically designs the menu offered. The optimal menu must serve a dual purpose: not only to effectively screen consumers based on their willingness to pay but also to sway the intermediary's information-disclosure choices for the seller's advantage. Our results uncover how information intermediaries can enhance product variety and consumer gains at the expense of sellers' profits and economic efficiency, thus underscoring their nuanced role in contemporary markets. 

Consider the tech industry as an example. Information intermediaries like tech review platforms or social media influencers often exhibit a bias toward highlighting high-end, cutting-edge technology products, steering the market towards luxury tech brands and influencing consumer choices towards more expensive options. This practice, aimed at maintaining a reputation for expertise, benefits sellers specializing in high-end products and reflects a significant influence on consumer choices due to information asymmetries. However, a potential shift in this bias, driven by various factors such as a desire to appeal to a broader audience or market demands for diverse product reviews, could lead these platforms to feature a wider range of products, including mid-range and budget-friendly options. This diversification in product coverage would lead to a more balanced market where products of various quality levels gain visibility. Such a shift would compel sellers, especially those in the high-end segment, to broaden their offerings to include more affordable options, recognizing a newfound opportunity to attract consumer attention. For consumers, this reduced bias results in exposure to a broader spectrum of products, enabling them to make more informed decisions that align better with their individual needs and budget constraints, thereby reshaping consumer behavior and market outcomes.

Such information intermediaries are ubiquitous in today's marketplace. While their importance is amplified in the digital economy, their influence is much broader. For instance, when a government entity or a corporation seeks new technology or services, such as military technology or innovative employee training to boost productivity, assessing the true value of such services or products can be challenging for decision-makers like politicians or CEOs. Here, specialized advisors, such as military generals or company department heads, provide guidance based on their expertise. Committed to optimizing returns, these advisors may also favor superior quality options. A critical question is how supply will react to these information frictions. To illustrate further, consider the role of environmental agencies or organizations. These entities provide insights into the production processes of goods, aiming to navigate consumers—who may grasp the importance of environmental sustainability but at the same time struggle to evaluate the precise impact of their buying decisions—toward greener choices. This raises certain questions: Will the available product selections be of higher quality, meaning more environmentally sound? And what strategies should sellers employ to optimally refine their offerings? 

In our model, a seller (she) offers goods of heterogeneous quality to screen different types of a consumer (he). In a standard screening framework, as in \cite{mussa_rosen_1978} or \cite{maskin_riley_1984}, the buyer's taste for quality is his private information. The novelty of our setup is that we assume that the buyer's valuation is unknown to both the buyer and the seller at the beginning of the game. Instead, the consumer has access to an information intermediary (he) who provides information regarding the buyer's valuation, \textit{after} the seller has already posted the menu of available products. The intermediary's preferences, as suggested by the aforementioned examples, encompass both the consumer's surplus and the quality of the products. The information provided by the intermediary facilitates consumer learning about individual preferences or, more naturally, about the product characteristics that influence the compatibility between the consumer and the product. Our primary focus lies in understanding how the intermediary's presence impacts market outcomes. To do so, we characterize the optimal finite-item menu that the seller can offer.\footnote{The restriction to finite-item menus is for tractability. This restriction is without loss of optimality when the intermediary's bias is sufficiently large.} The composition of this menu is intimately tied to the intermediary’s goals; as such, it varies according to the intermediary's objective function. Naturally, we then ask how the profit-maximizing menu of offered products is adjusted in response to changes in the intermediary's objective function. In other words, we investigate how it evolves as the intermediary's emphasis shifts between product quality and consumer surplus. We leverage our results to analyze the implications of the change in supply's response to market outcomes of interest, specifically average quality, profits, consumer surplus, distortions, and overall economic efficiency.

Because the intermediary provides information after the monopolist announces the menu, the monopolist has the capacity to shape the intermediary's decisions regarding information disclosure through the menu's design. This, in turn, yields an additional constraint on the monopolist's problem. Since the seller knows what information the intermediary will provide in equilibrium, we interpret the additional constraint as an obedience constraint on the intermediary's behavior: the seller chooses, in addition to the menu of goods, the information the buyer will receive as well. This choice must be optimal for the intermediary in the sense that the intermediary must have no incentive to deviate and provide information in a different way than the one suggested by the monopolist. From a methodological perspective, we model the information provision stage as a Bayesian persuasion problem. Specifically, the intermediary acts as a sender, responsible for selecting and committing to an information structure aimed at influencing the consumer's decision. The consumer, who takes on the role of the receiver, observes the realization of this signal before making a choice from the menu. Thus, our model incorporates a Bayesian persuasion problem as a constraint to the monopolist's problem.

Our primary theoretical contribution lies in the characterization of the optimal finite-item menu that the seller offers and its responsiveness to shifts in the intermediary's bias. We first establish that, under some additional assumptions on the prior distribution from which the consumer's value is drawn, if the intermediary's bias is sufficiently high, the intermediary is essentially redundant. The optimal menu and how information is provided to the consumer coincide with the case where the seller provides information directly to the consumer and there is no intermediary. This happens because the seller offers a single-item menu absent the intermediary and provides information so as to maximize the probability of trade. As a result, the consumer receives no surplus. Even if the intermediary is present, however, the dissonance between the intermediary's and the consumer's objectives is so pronounced that the intermediary essentially is better off by providing information to maximize the probability of trade instead of yielding the consumer a positive payoff, which is exactly the optimal way the seller would provide information. 

On the other hand, when the intermediary's bias is low, the obedience constraint on the intermediary's behavior binds and fully determines the posterior types of the consumer, given the menu the seller posts. The paper's main result is that as the intermediary assigns greater weight to consumer surplus relative to product quality, the optimal menu expands to include a growing number of products. This shift is a strategic response by the seller to improved consumer learning facilitated by a less biased intermediary. In particular, by increasing the number of product options, the monopolist can better tailor the menu to fit the demand of posterior buyer types. To see this, suppose that the seller's profit-maximizing menu is a single-item menu for a given level of the intermediary's bias and suppose that there is a decrease in the bias term. Then, the consumer's posterior value will be closer to their true value, and fewer consumers will purchase the high-quality item. Suppose, now, that the seller decides to include in the menu a second, lower-quality option. This will induce some consumers who do not purchase the high-quality product to purchase the new, more affordable option. Moreover, some consumers will switch from the high-quality product to the lower-quality one. The introduction of the new item will benefit the seller only when the decrease in the intermediary's bias is sufficiently high. Moreover, by the same logic, the seller must include a successively higher number of options in the menu as the bias keeps decreasing. In the limiting case where the intermediary's bias approaches zero so that the intermediary's and the consumer's preferences coincide, it is optimal for the seller to offer a continuum of items. This is because, in this case, the intermediary provides information to the consumer in a way that guarantees that the consumer does not make ex-post inefficient trading decisions. This can only happen when the intermediary induces the consumer to learn the product's true value. Thus, it is as if the consumer has private information about their value, as in the classic monopolistic screening model of \cite{mussa_rosen_1978}.

In terms of market outcomes of interest, while the presence of intermediaries can incentivize suppliers to broaden their product offerings, thereby augmenting consumer benefits, there is a reduction in overall economic efficiency under the intermediary's presence compared to scenarios where sellers directly provide information to consumers.\footnote{For consumers, the increased product variety is beneficial. Especially in the contemporary digital economy, the benefits of expanded product variety in electronic markets to consumer welfare are widely acknowledged. While the competitive nature of these markets leads to consumer gains, such as reduced average selling prices, thereby boosting consumer surplus, \cite{bryn_hu_smith_2003} highlight that the surge in product variety can be an even more significant driver for these gains. This variety increase is facilitated by online retailers' capacity to efficiently catalog and recommend a vast array of products. Our framework seeks to encapsulate the role of these platforms as information intermediaries, emphasizing their ability to guide consumers to appropriate products.} This happens because a decrease in the intermediary's bias reduces the seller's profits. While introducing a wider range of products mitigates the profit loss, it does not fully reverse it. The reason is the convexity of the cost function associated with production. The expected cost of offering a menu with a high number of options is sufficiently higher than the one if the menu features few options. Thus, profits decrease as the intermediary's bias does, and this reduction dominates the gains in consumer payoffs. Moreover, while the consumer's payoff in the presence of the intermediary is higher relative to the case where the seller provides information to the consumer directly, it exhibits a non-monotonic relationship with the bias term. In essence, consumers favor an intermediary with a very low bias over one with a very high bias. Yet, when the difference in bias is not exceedingly substantial, the highly biased intermediary may be preferred. This observation underscores the nuanced forces at play. As a result, efficiency in the economy is also non-monotonic in the intermediary's bias with a decreasing trend. 

The above results have significant implications for designing policies to facilitate consumer learning. The rationale behind such interventions hinges on the impact of information on consumer decision-making, particularly in preventing choice errors where buyers fail to select the product that best suits their needs. Consequently, in situations where the qualities and prices of products are fixed, increasing the availability of information to consumers will enhance their welfare. However, this argument may not hold when sellers are allowed to respond to this information provision. In such cases, a seller's ability to tailor products and prices to different consumer types introduces a non-monotonic change in consumer welfare, as we discussed. Consequently, while some degree of information is undoubtedly better than none, even a slight alteration in the information made available to buyers might reduce consumer payoffs from the interaction. In other words, the relationship between consumer surplus and information availability becomes more nuanced when intermediaries are present, and supply can adapt to consumer learning.

Furthermore, we examine the consequences arising from the seller's constraints in expanding their product variety to its optimal level imposed by external factors. Recognizing that real-world firms often face such restrictions makes this analysis vital. Such constraints introduce additional inefficiencies. The seller, in this context, cannot counterbalance the profit decline resulting from the intermediary's reduced bias. Consequently, consumers face adverse outcomes, as, in essence, a broader product range enhances the alignment between buyers and their desired goods. Therefore, these limitations are detrimental both to monopolist profits and consumer surplus and, thus, lead to a diminished overall total surplus.

Our results highlight the importance of recognizing the information intermediaries' role when assessing market inefficiencies. In our model, the average quality level is higher than what it would be under standard screening models yet lower than if consumers relied solely on direct seller-to-consumer information. Therefore, if a researcher overlooks the intermediary's role, they might inaccurately estimate the actual demand--either underestimating or overestimating it. This oversight would lead them to wrongly judge the consumer's preference for quality and, consequently, the extent of quality distortions.

\subsection*{Related Literature}
Our paper is related to several strands of literature. Our screening setting is in the spirit of \cite{mussa_rosen_1978} and \cite{maskin_riley_1984}. The key distinction is that information about the consumer's preferences for quality or product characteristics that affect the match with the good is provided by an intermediary rather than the consumer having private information.

Our work closely aligns with \cite{bergemann_et_al_2022}, who explore optimal pricing through both mechanism and information design. Their focus is on a seller directly providing information to the consumer. In contrast, we consider the case where such information is provided by an intermediary. In particular, the problem \cite{bergemann_et_al_2022} solve is a relaxed version of our problem since, in our framework, there is an additional constraint on the monopolist's problem that captures the required obedience of the intermediary. We demonstrate that when this intermediary has a significant bias toward high-quality products, the seller's optimal menu coincides with \cite{bergemann_et_al_2022}'s. Yet, if this bias is small, a broader product range is optimal. We further characterize the exact number of items that the optimal menu must feature and how the intermediary's presence affects market outcomes.

A different approach to capturing information frictions in a monopolistic screening environment is pursued in \cite{thereze_2022} and \cite{mensch_ravid_2022}. These papers employ a rational inattention framework, assuming consumers can access relevant information at a cost. We provide an in-depth discussion of how the two approaches differ and complement each other in Section 6. \cite{mensch_2022} adopts a similar rational inattention approach focusing on the sale of a fixed-quality product to single or multiple buyers.

More generally, our work complements a literature examining the interplay of information provision and mechanism design. \cite{roessler_szentes_2017} model a scenario where buyers obtain information before sellers present their offerings. In a related vein, \cite{roessler_ravid_szentes_2022} study a setting where buyer learning and seller pricing happen concurrently. A distinguishing factor between their works and ours lies in the sequence of events. In our model, the monopolist sets the menu before the intermediary imparts information to the consumer. This sequence is pivotal. More importantly, while \cite{roessler_szentes_2017} address optimal buyer learning--a situation we mimic when the intermediary has no bias--\cite{roessler_ravid_szentes_2022} presume consumers pay for information, akin to \cite{thereze_2022} and \cite{mensch_ravid_2022}, while we assume that the intermediary provides information to the consumer for free.

\cite{malenko_tsoy_2019} study an auction framework where biased advisors provide information to uninformed bidders in a cheap-talk communication framework. They use the same functional form to capture the bias term of the expert as we do. They show that dynamic mechanisms dominate static ones and derive the optimal selling mechanism in many scenarios. Our paper differs in that the ``expert" has commitment power relative to \cite{malenko_tsoy_2019} and, moreover, we study the problem of selling multiple quality-differentiated products to a single consumer instead of the allocation of a single object to many buyers.

From a methodological standpoint, our approach leans heavily on the Bayesian Persuasion literature. We leverage findings related to dual price functions and bi-pooling distributions derived by \cite{d_martini_2019}, \cite{deniz_kovac_2019}, \cite{ksm_2021} and \cite{arieli_et_all_2023}. A more detailed discussion of the methodological connections is provided in Section 3.2 when the necessary notions are introduced.

\subsection*{Organization of the Paper}
The rest of the paper proceeds as follows. Section 2 introduces the model. Section 3 introduces necessary tools from the Bayesian Persuasion literature and establishes some preliminary results. Section 4 provides conditions under which the optimal mechanism in our environment coincides with or differs from the one in a setting in which the monopolist provides the information to the consumer. Section 5 characterizes the optimal finite-item menu and establishes the necessity of product variety expansion as the intermediary bias decreases. Section 6 specializes the analysis to the Uniform-Quadratic setting to further shed light on the structure of the optimal menu and its comparative statics. Section 7 connects our approach with a different strand of the literature, which studies the effect of supply responses to information frictions using an information acquisition framework. Section 8 concludes. All proofs and supporting calculations are in the Appendix.
\medskip

\section{Model}

A monopolist (she) sells goods of varying quality to a potential buyer (he). The value of the good to the buyer, his type, is given by a random variable $\theta\in [0, 1]$ which is distributed according to the cumulative distribution function (CDF) $F_0$ with density $f_0$, positive everywhere in the support of $F_0$.\footnote{The fact that the support of the buyer's valuation $\theta$ is the $[0,1]$ interval is a normalization and is, of course, withouta loss. The analysis goes through unchanged for $\theta\in[\underline{\theta},\Bar{\theta}]$ where $\underline{\theta}\geq 0$ and $\Bar{\theta} <\infty$.} Throughout the paper, we assume that $F_0$ has an increasing hazard rate and, thus, is Myerson regular. The game starts with the monopolist offering the buyer a contract $M$ consisting of pairs of menu items $(q,t)\in [0,\bar{q}] \times \mathbb{R}_+$. Each menu item corresponds to a transfer $t$ to be paid to the monopolist by the buyer and the quality $q$ of the product the buyer gets in exchange. We assume the utility of the buyer is quasi-linear. That is, given the buyer's type $\theta$, his utility from the menu item $(q,t)$ is given by
\begin{equation}
    U^B (\theta,q,t) = \theta q-t
\end{equation}

The monopolist's cost of providing quality $q$ is given by $c(q)$ where the function $c:\mathbb{R}_+\to\mathbb{R}_+$ is assumed to be strictly increasing, continuously differentiable, and strictly convex. We assume that the buyer has the option of not buying anything, that is, we require that $(0,0)$ is included in the menu.

Neither the monopolist nor the buyer knows the value of $\theta$. However, the buyer has access to an informational intermediary (he) who can provide information to the buyer about $\theta$. Specifically, after observing the posted menu, the intermediary can pick and commit to an information structure $s:[0,1]\to\Delta ([0,1])$ whose realization is observed by the buyer before his choice of a menu item. We assume that the intermediary's payoff if the buyer chooses the item $(q,t)$ is given by
\begin{equation}
    U^I(\theta,q,t)= (\theta +b)q-t 
\end{equation} 
where the parameter $b$ captures the intermediary's bias toward higher-quality products and is commonly known to all players. Alternatively, one can think of the intermediary's payoff as placing weight one on the consumer's payoff and weight $b$ on the quality of the product. The case where $b=0$ corresponds to the instance where the intermediary's preferences are fully aligned with the buyer's. This case can also be interpreted as the buyer costlessly acquiring information himself.

 Given the realization $s$, the buyer's expected value is denoted by 
\begin{equation}
w:= \EE (\theta|s)
\end{equation}
Our assumptions on the buyer's and advisor's utilities imply that their expected payoff from any menu item depends on the posterior mean $w$. Therefore, the marginal distribution of $w$ pins down the buyer’s and intermediary's expected trade surplus from any menu. It also determines the probability the buyer purchases any item, which, in turn, is sufficient for calculating the monopolist’s profits. In other words, trade outcomes depend only on the marginal distribution of the buyer’s posterior mean, and so we identify each signal with the CDF of this marginal, which we denote by $G$, with support $\supp G$.

Throughout the paper, we assume that when the buyer is indifferent between two items, he purchases the higher-quality one and that if he is indifferent between purchasing an item or not participating in the mechanism, he chooses to participate. That is, we assume that the buyer's indifference is broken in favor of the seller. Since the intermediary is biased towards higher-quality items, this implies that the buyer's indifference is also broken in favor of the intermediary. This, in turn, implies that the intermediary's problem always has a solution (see \cite{arieli_et_all_2023}).

The timing of the game is as follows: First, the monopolist posts the menu $M$. Then, the advisor commits to the information structure $s$, and nature draws $\theta$. Finally, the buyer observes the signal realization, chooses an item from the menu, and payoffs accrue.

This timing of the game implies that the buyer's interim expected payoff depends only on his posterior mean $w$. The revelation principle applies and we can restrict attention to direct mechanisms. We can describe these mechanisms with two mappings, $q:[0,1]\to [0,\bar{q}]$ and $t:[0,1]\to \mathbb{R}_+$ where $q(w)$ and $t(w)$ correspond to the quality and transfer pair that a buyer with posterior mean $w$ chooses. The mechanism must satisfy the standard individual rationality and incentive compatibility constraints:
\begin{equation}
    wq(w)-t(w) \geq wq(w')-t(w')  \ \text{for all} \ w,w'\in \supp G
\end{equation}
\begin{equation}
      wq(w)-t(w) \geq 0 \ \text{for all} \ w\in \supp G
\end{equation}

Given $q(w)$ and $t(w)$ the advisor chooses the information structure $s$ which induces the distribution $G$ over posterior means $w$. As mentioned, we work directly with $G$. As is well known in the Bayesian Persuasion literature, $G$ is the CDF of the marginal distribution of the buyer’s posterior mean for some information structure if and only if it is a mean-preserving contraction of the prior $F_0$ or, equivalently, if and only if $F_0$ is a mean-preserving spread of $G$. Let $\mathcal{F}$ denote the set of CDFs over the interval $[0,1]$. Recall that $F\in\mathcal{F}$ is a mean-preserving spread of $G$ if and only if 
\begin{equation*}
    I_G(\theta):= \int_0^{\theta} (F_0-G)(x)dx \geq 0 \ \text{for all} \ \theta\in [0,1] \ \text{with equality at} \ \theta =1
\end{equation*}
Therefore, we define the set of mean-preserving contractions of $F_0$, which corresponds to the set of feasible distributions over posterior means by
\begin{equation}
MPC(F_0)=\{G\in\mathcal{F}: I_G(\theta) \geq 0 \ \text{for all} \ \theta \ \text{and} \ I_G(1)=0\}
\end{equation}

The intermediary's problem given the menu $(q,t)$ can then be written as

\begin{align*}
     \max_{G\in MPC(F_0)} \int_0^1 \left[(w+b)q(w') -t(w')\right] dG(w) \\
    \text{s.t} \ w'\in\argmax_{\hat{w}\in \supp G\cup \{n\}} \{w'q(\hat{w})-t(\hat{w}\}
\end{align*}
where $\{n\}$ denotes the option of the buyer to not participate in the mechanism. However, since we assume $(q,t)$ is IC and IR, the constraint in the advisor's problem is redundant. 

We define 
\begin{equation}
    u^I(w):= (w+b)q(w)-t(w)
\end{equation}
to be the advisor's indirect utility from inducing posterior mean $w$. The intermediary's problem is a standard linear persuasion problem in which he acts as the sender and the buyer is the receiver. The caveat is that the shape of the advisor's indirect utility function depends endogenously on the menu $(q,t)$ that the seller posts.

Note that $G$ may have gaps, but it is without loss of generality to assume that $q$ is defined on the whole domain $[0,1]$. With this in mind, we refer to a pair of a distribution and a menu $((q,t),G)$ as an outcome. The seller's problem can then be written as:
\begin{align*}
\max_{(q(w),t(w)), G\in MPC(F)} \int_0^1 \left[t(w)-c(q(w))\right] dG(w)\\
    \text{s.t.} \ wq(w)-t(w) \geq wq(w')-t(w')  \ \text{for all} \ w\in[0,1] \tag{B-IC} \\
    wq(w)-t(w) \geq 0 \ \text{for all} \ w\in[0,1] \tag{B-IR} \\
G\in \argmax_{G\in MPC(F_0)} \int_0^1 \left[(w+b)q(w) -t(w)\right] dG(w) \tag{I-OB}
\end{align*}

 The seller's problem is a standard monopolistic screening problem with one additional constraint: the intermediary chooses the distribution over the buyer's posterior means. We can view this constraint as an obedience requirement for the advisor. If the seller posts the menu and suggests a distribution $G$ to the advisor, then he must have no incentive to deviate and choose a different distribution.
\medskip

\section{Preliminary Analysis}

\subsection{Existence of the Optimum}

We begin by showing the existence of a solution to the monopolist's problem. This can be established due to the compactness of the menu space we are considering. Additionally, the intermediary's problem is a linear persuasion problem, which is well-known to have a solution when the sender's utility function exhibits upper-semi-continuity.

In our framework, the upper-semi-continuity of the intermediary's utility is derived from the assumption that the buyer engages in the mechanism when indifferent between participation and non-participation. Furthermore, the buyer opts for the higher quality item when faced with indifference between two items.

\begin{proposition}\label{existence}
    A solution to the monopolist's problem exists.
\end{proposition}

\subsection{Reformulating the Monopolist's Problem}
We can view the monopolist's problem as having two components: she suggests a distribution over posterior means to the intermediary, which must be incentive compatible, and she chooses a mechanism that specifies the quality exchanged and the transfer for any reported buyer type. Notice that the distribution $G$ might not be continuous or atomless. Moreover, it may have finite support. However, it is without loss of generality to assume that the allocation rule $q$ and transfers $t$ are defined on the whole domain $[0,1]$. Then, standard arguments deliver that IC and IR constraints are satisfied if and only if $q$ is increasing\footnote{Throughout the paper the term increasing means weakly increasing. The same holds for the terms convex, concave, etc.} and the envelope formula is satisfied:
\begin{equation}
    t(w)=wq(w)-\int_0^w q(s) ds -u_0
\end{equation}
where $u_0$ is the utility the lowest type receives, which optimally equals zero. 

With this, we can re-write the advisor's indirect utility function as

\begin{equation}
    u^I(w)= bq(w) +\int_0^w q(s) ds
\end{equation}

and restate the monopolist's problem as

\begin{align*}
\max_{q(w), G\in MPC(F_0)} \int_0^1 \left[wq(w)-\int_0^w q(s) ds-c(q(w))\right] dG(w)\\
\text{s.t.} G\in\argmax_{G\in MPC(F_0)} \int_0^1 \left[bq(w) + \int_0^w q(s) ds\right] dG(w)
\end{align*}

We, therefore, look for the optimal mechanism among incentive compatible outcomes $(q,G)$ such that, given $q(w)$, $G$ satisfies the obedience requirement for the advisor, and, given $G$, $q$ is incentive compatible for the buyer.

\subsubsection{The Dual Price Function and Bi-Pooling Distributions}
We proceed to introduce two objects that will be useful throughout the analysis. Consider the standard linear persuasion problem 
    \begin{align*}
 \max_{G\in MPC(F_0)} \int_0^1 u(x) dG(x)
    \end{align*}

\begin{definition}\label{price_function}
    A function $p:[0,1]\to\mathbb{R}$ is called a \textit{(dual) price function} of $G$ if
    \begin{itemize}
        \item[(i)] $p(x)\geq u(x)$ for every $x\in[0,1]$.
        \item[(ii)] $p$ is convex.
        \item[(iii)] $\supp G \subseteq \{x\in [0,1]: p(x)=u(x)\}$.
        \item[(iv)] $\int_0^1 p(x) dG(x) =\int_0^1 p(x) dF(x)$
    \end{itemize}
\end{definition}

\cite{d_martini_2019} prove that if there exists a price function $p$ of $G$ then $G$ solves the persuasion problem and, if
$u$ satisfies certain  conditions, then for any optimal distribution $G$ such a price function exists. In principle, the conditions required by \cite{d_martini_2019} do not necessarily hold in our environment. However, \cite{deniz_kovac_2019} generalize the result of \cite{d_martini_2019} and provide weaker conditions under which the dual price function exits. Specifically, they require that there exists and $\epsilon>0$ such that $u(\cdot)$ is Lipschitz continuous on $[0,\epsilon]$ and $[1-\epsilon, 1]$. Recall that the intermediary's utility is given by $u^I(w)= bq(w) +\int_0^w q(s) ds$. It follows immediately that $u^I(\cdot)$ cannot explode at the bounds of $[0,1]$ since $q(\cdot)$ does not, and, therefore, the dual price function always exists in our problem. Next, we define the class of distributions in which the solution to the intermediary's problem belongs to.

\begin{definition}\label{bi_pooling}
    A distribution $G\in MPC(F_0)$ is a \textit{bi-pooling} distribution if it partitions $[0,1]$ in disjoint intervals such that in each interval:
    \begin{itemize}
        \item[(i)] Either types are fully disclosed, in which case $G=F_0$, or
        \item[(ii)] Types are pooled and $G$ has an atom equal to the measure of the interval at the mean of the interval, or 
        \item[(ii)] Types are bi-pooled, that is, $G$ has two atoms $y_1$ and $y_2$ in this interval with weights $\alpha$ and $(1-\alpha)$, such that $\alpha y_1 +(1-\alpha) y_2= \hat{y}$ where $\hat{y}$ is the mean of $F_0$ when restricted to this interval.
    \end{itemize}
\end{definition}

\cite{arieli_et_all_2023} prove that if $u(\cdot)$ is upper semicontinuous, there is always a solution $G$ to the linear persuasion problem that belongs to the class of bi-pooling distributions. This, in our framework, means that for every increasing allocation function $q$, there is a solution $G_q$ of the intermediary's problem which is a bi-pooling distribution. Finally, if distribution $G\in MPC(F_0)$ partitions $[0,1]$ into intervals and pools types in each interval at its mean, it is called a \textit{monotone pooling} distribution.

The two examples that follow illustrate how the dual price function is used to solve Bayesian Persuasion problems in our framework.

\begin{example}
    Suppose the prior distribution $F_0$ is the standard Uniform distribution and the cost function is given by $c(q)=q^2/2$. It is well known that, in this case, the solution to the standard monopolistic screening problem where the buyer knows his valuation is given by $q^{MR}(\theta) =(2\theta-1)_+$.\footnote{We use the notation $(x)_+=\max(x,0)$.} Then the intermediary's utility function given $q^{MR}$ is $u^I(w)=0$ for $w<1/2$ and
    \begin{equation*}
        u^I(w)=w^2+(2b-1)w+\frac{1}{4}-b
    \end{equation*}
    for $w\in[1/2,1]$.
It is straightforward to see that $u^I(\cdot)$ is continuous everywhere and strictly convex in $[1/2,1]$. Therefore, the dual price function can be chosen to be $p(w)=0$ for $w\in[0,1/2)$ and $p(w)=u^I(w)$ for $w\in[1/2,1]$. Thus, the solution to the intermediary's problem is to fully disclose types in $[1/2,1]$ and either pool or fully disclose types in $[0,1/2)$. Therefore, the outcome $(q^{MR},F_0)$ is IC.
\end{example}

\begin{example}
    Suppose now that the prior distribution is still the standard uniform but the cost function is given by $c(q)=q^3/3$. In this case, $q^{MR}(\theta)=\sqrt{(2\theta-1)_+}$. Again, a straightforward calculation yields that $u^I(w)=0$ if $w\in[0,1/2)$ and 
    \begin{equation*}
     u^I(w)=b (2\theta-1)^{\frac{1}{2}}+\frac{1}{3} (2\theta-1)^{\frac{3}{2}}   
    \end{equation*}
    if $w\in[1/2,1]$.
Now, $u^I(\cdot)$ is concave up to a point and then becomes convex. The solution to the intermediary's problem is to pool types in some interval $[v_1,v_2]$ and assign an atom of size $(v_2-v_1)$ at $m=\EE(\theta|v_1\leq \theta\leq v_2)$. The dual price function in this case is given by

\[
p(w) = \begin{cases}
0 & \text{if}~w \in [0, v_1) \\
\frac{u^I(v_2)-u^I(v_1)}{v_2-v_1}(w-v_1) & \text{if}~ w\in [v_1,v_2) \\
u^I(w) & \text{if}~ w\in[v_2,1]
\end{cases},
\]

In Figure \ref{figure1}, we plot the intermediary's utility and the dual price function in this example. As expected by Proposition 3 in \cite{arieli_et_all_2023}, the price function is equal to the intermediary's utility in $[v_2,1]$ where full disclosure happens and affine and tangential to $u^I$ at the mean of $[v_1,v_2]$ where pooling takes place.
    \begin{figure}[htbp]
    \centering
    \includegraphics[scale=0.3]{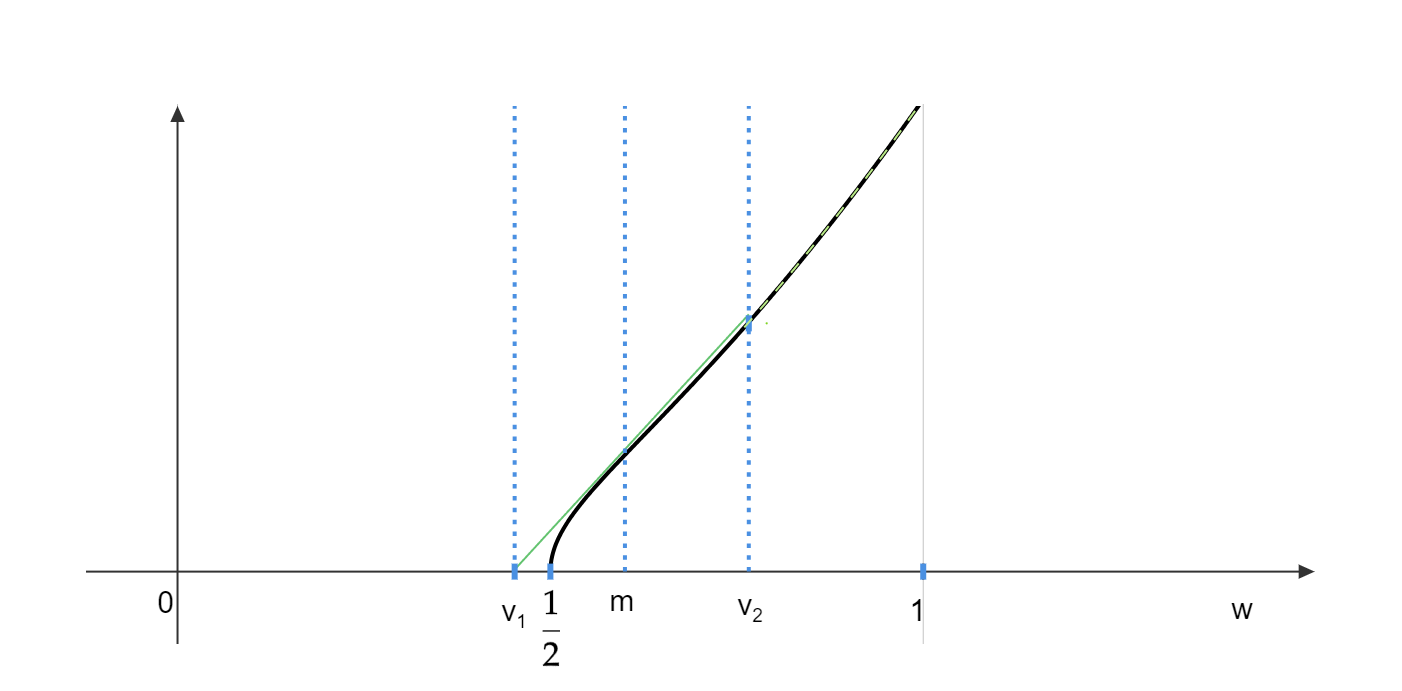}
    \caption{Intermediary's utility (in black) and the dual price function (solid and dashed green) in Example 2.}
    \label{figure1}
\end{figure}

 It follows that the outcome $(q^{MR}, F_0)$ is not IC. Still, the quality mapping $q^{MR}$ can be implemented by the distribution $G_{q^{MR}}$ that takes the aforementioned form. Note that this case features \textit{lower-censorship} in the sense of \cite{kolotolin_et_al_2022}: the intermediary pools the states in $[v_1,v_2]$ and reveals the states above $v_2$.

\end{example}

\subsection{Benchmark: Fully-Aligned Preferences}
To start our analysis, we first consider the benchmark case where $b=0$. Now, the intermediary shares the same preferences with the buyer, a situation which can be interpreted as the buyer costlessly choosing the distribution $G$ himself. Let $q^{MR}$ denote the Mussa-Rosen menu that would be optimal if the buyer knew his value $\theta$. The following result is intuitive.

\begin{proposition}
    In the case where $b=0$, the outcome $(q^{MR}, F_0)$  is the solution to the monopolist's problem. 
\end{proposition}

Proposition 2 says that the monopolist's optimal outcome is the Mussa-Rosen menu together with the prior distribution. This follows because, in this case, the intermediary provides information to ensure that the buyer makes efficient ex-post trading decisions. Thus, the seller's profit under any outcome $(q,G)$ must be equal to her profit under outcome $(q,F_0)$. 

Formally, given $q(w)$, the intermediary's problem is
\begin{equation}
    \max_{G\in MPC(F)} \int_0^1 \left[\int_0^w q(s)ds\right]  dG(w)    
\end{equation}
where $u^b(w):=\int_0^w q(s)ds$ is the buyer's rent. Since $q(w)$ must be increasing in order to satisfy IC, it follows that $u^b(\cdot)$ is convex. Thus, choosing signal $F_0$, that is, fully revealing the buyer's true value is always a solution to the intermediary's problem. This yields the intermediary a payoff of 
\begin{equation}
    U^I_{F_0} =  \int_0^1 \left[\int_0^w q(s)ds\right]  dF_0(w)= \int_0^1 q(w) (1-F_0(w))dw
\end{equation}

Choosing any other signal $G$ yields a payoff of $U^I_G= \int_0^1 q(w) (1-G(w))dw$. The extra benefit from switching from signal $G$ to the prior $F_0$ is, then, given by
\begin{equation*}
  U^I_{F_0} -U^I_G= \int_0^1 q(w) [F(w)-G(w)]dw \geq 0
\end{equation*}
by the fact that $I_G(\theta)\geq 0$ for all $\theta\in [0,1]$ and $q(\cdot)\geq 0$. In this case, the intermediary can only lose by choosing a less informative distribution so a signal $G$ is a best response to menu $q$ if and only if 
\begin{equation*}
    \int_0^1 q(w) [F(w)-G(w)]dw = 0
\end{equation*}

Thus, in the case where $b=0$ the seller's problem can be simplified to

\begin{align*}
\max_{q(w), ~G\in MPC(F)} \int_0^1 \left[wq(w)-\int_0^w q(s) ds-c(q(w))\right] dG(w)\\
\text{s.t.}  \int_0^1 q(w) [F_0(w)-G(w)]dw = 0
\end{align*}

For a signal $G$ to satisfy obedience, it must be that it yields the buyer the same expected payoff as the fully revealing signal $F_0$. Equivalently, it must be the case that the buyer makes efficient ex-post trading decisions as to what item to choose from the menu (if any). This implies that, given $q$, the trading behavior of the buyer must be the same under $G$ and $F_0$.  This, in turn, yields that whatever the profit the seller can achieve under $(q, G)$, she can also achieve under $(q, F_0)$. Since the Mussa-Rosen mechanism is optimal when the buyer's valuation is his private information, it is the best the seller can achieve when $b=0$ as well. 

\section{Highly Biased Intermediary}
We now consider the case $b>0$, that is, the intermediary exhibits some bias towards high-quality products. We will see that if this bias is sufficiently high, then, under some additional assumptions, the overall optimal mechanism is a single-item menu, and in particular, coincides with the one that would be optimal if the seller were providing information to the consumer directly. \cite{bergemann_et_al_2022} prove that it is without loss to focus on finite-item menus when the monopolist designs the information structures as well as the menu. Let $(q^*_{BHM}, G^*_{BHM})$ denote the monopolist's optimal outcome in this case. A natural question is under what conditions on the bias term this outcome remains optimal despite the intermediary's presence. Clearly, this will be the case only if the (I-OB) constraint is satisfied. We provide a sufficient condition on the bias term for this to be the case whenever the prior distribution is such that $q^*_{BHM}$ features a single item. We also provide a threshold on the intermediary's bias such that if the bias is smaller than this threshold, $(q^*_{BHM}, G^*_{BHM})$ for sure violates (I-OB) and, thus, cannot be optimal. 

\subsection{Single-Item}

Suppose that the marginal cost is convex, that is, $c'''(q)\geq 0$. Then, under the following additional assumption on the prior distribution, the overall optimal mechanism if the seller were providing information directly to the consumer, is a single-item menu as \cite{bergemann_et_al_2022} show. 
\begin{assumption}\label{single_item}

 The prior distribution $F_0$ has density $f_0$ that satisfies:
\begin{equation*}
    f_0'(\theta)<0 \Rightarrow f_0''(\theta)\leq 0 
\end{equation*}

\end{assumption}

Note that Assumption \ref{single_item} is satisfied by any distribution with an increasing or concave density. Now, the seller's problem amounts to finding a value $v^*$ above which types are pooled and a quality $q^*$ that is offered to type $w^*:=\EE(\theta|v^*\leq \theta \leq 1)$. The price is then given by $p^*=w^* q^*$ so that no rent is left to the posterior type $w^*$. Formally, the monopolist's problem is given by

\begin{align*}
    \max_{(q,~v)}  ~(1-F_0(v))\left(wq-c(q)\right)
\end{align*}
where $w:=\EE(\theta| v\leq \theta\leq 1)$.
The optimal value $v^*$ and quality $q^*$ can then be found by solving the first-order conditions:
\begin{align*}
    w^*=c'(q^*) \ \ \text{and} \ \ v^*q^*=c(q^*) \tag{OSIM}
\end{align*}

\begin{example} \label{mainexample}
    Consider the case where $F_0$ is the standard Uniform distribution and the cost is given by $c(q)=q^2/2$. Then, $v^*=1/3$, $q^*=2/3$, $p^*=4/9$ and types are pooled in the intervals $[0,1/3]$ and $[1/3,1]$. 
\end{example}

Our first result establishes that the monopolist's optimal menu is the one that would be optimal if she were controlling the information directly, if and only if, the intermediary's bias is higher than a threshold.

\begin{proposition}\label{high-b}
    Suppose that $c'''(q)\geq 0$ and that Assumption 1 holds. Then, there exists a value $b^*$ such that the monopolist's optimal outcome is the single-item menu $q^*_{BHM}$ together with distribution $G^*_{BHM}$ as specified by conditions (OSIM), if and only if $b\geq b^*$. 
\end{proposition}

Proposition \ref{high-b} says that in the case where the prior distribution is such that the solution to the monopolist's problem absent the intermediary is a single-item menu, there is a threshold $b^*$ such that if the intermediary's bias is higher than this threshold, the (I-OB) constraint in the monopolist's problem does not bind. Thus, she can achieve the highest profit possible: the one she would attain if she provided information to the buyer. This will not be the case when the intermediary's bias is smaller than $b^*$.

\addtocounter{example}{-1}
\NextLinkTarget{example.continued}
\begin{example}[Continued]
In the framework of Example \ref{mainexample}, in order for the single-item menu derived before to remain optimal in the intermediary's presence, it must be the case that upon the seller posting $(q^*,p^*)$, the solution to the intermediary's problem is to pool types in the same way as the seller optimally would. As illustrated in Figure 2, this is the case only if $b\geq p^*/q^* -v^*=1/3$, so that $w^*-b=p^*/q^*-b \leq v^*$. In Figure 2, the solid blue line corresponds to the intermediary's utility, which is given by

\[
u^I(w) = \begin{cases}
0 & \text{if}~w \in [0, \frac{2}{3}) \\
(w+b)\frac{2}{3}-\frac{4}{9} & \text{if}~ w\in [\frac{2}{3}, 1] 
\end{cases},
\]
and the red dashed and solid line corresponds to the dual price function.

\begin{figure}[htbp]
    \centering
    \includegraphics[scale=0.3]{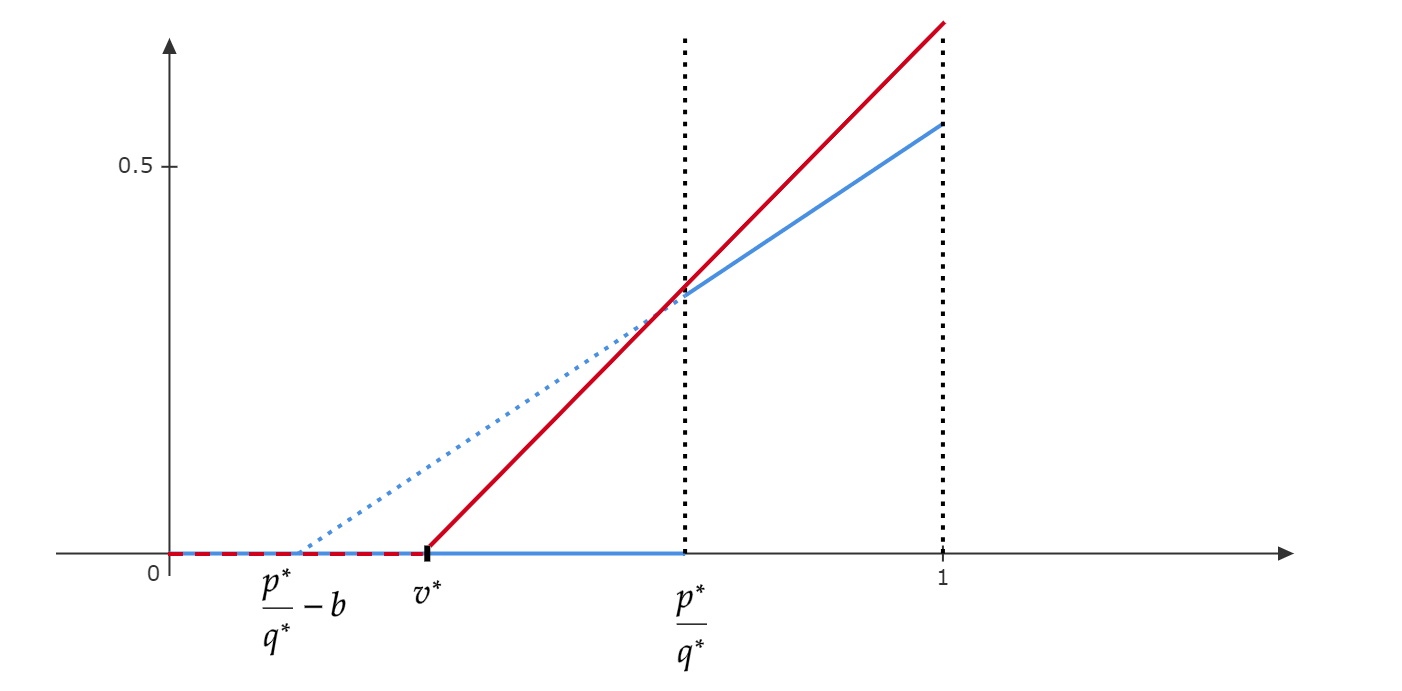}
    \caption{Single-item menu $q^*_{BHM}$ is optimal only if $b\geq 1/3$}
    \label{figure2}
\end{figure}

\end{example}

\subsection{Many Items}

If the optimal menu in the absence of the intermediary features more than one item, things are more complicated. Our next result provides a necessary condition that needs to hold in order for $(q^*_{BHM},G^*_{BHM})$ to satisfy the (I-OB) constraint. 

To this end, suppose that the optimal finite-item menu $q^*_{BHM}$ features $N$ items. Clearly, it must be the case that the downward buyer IC constraints must hold with equality since the seller can always adjust the distribution over posterior means so that no posterior buyer type gets more rent than what is absolutely necessary. Let $_v1=t_1/q_1$, $v_i=(t_i-t_{i-1})/(q_i-q_{i-1})$ for $i=2,...,N$. Then, the monopolist chooses $N$ values $\{v_i\}_{i=1,\dots,N}$ and qualities $\{q_i\}_{i=1,\dots,N}$ to solve:
\begin{align*}
    \max_{\{v_i\}_{i=1,\dots,N},\{q_i\}_{i=1,\dots,N}} \sum_{i=1}^{N-1}\left[ \left( F_0(v_{i+1})-F_0(v_i)\right)(t_i-c(q_i)) \right]+\left( 1-F_0(v_N)\right) (t_N-c(q_N))
\end{align*}
where $t_1=w_1 q_1$, $t_i=t_{i-1}+w_i (q_i-q_{i-1})$ for $i=2,..,N$ and
\begin{align*}
     w_0= & \EE(\theta|0\leq \theta\leq v_1)\\
    w_1 & = t_1/q_1=\EE(\theta|v_1\leq\theta\leq v_2)\\
    w_i & =(t_i-t_{i-1})/(q_i-q_{i-1})=\EE(\theta|v_i\leq \theta\leq v_{i+1}) \ \ \text{for} \ \ i=2,...,N-1\\
    w_N & = \frac{t_N-t_{N-1}}{q_N-q_{N-1}}=\EE(\theta|v_N\leq \theta\leq 1)
    \end{align*}

Let the  values $\{v^*_i\}_{i=1,\dots,N}$ and the qualities $\{q^*_i\}_{i=1,\dots,N}$ be the solution to the monopolist's problem. The optimal outcome is then given by $q^*_{BHM}=\{q^*_i\}_{i=1,\dots,N}$ and $G^*_{BHM}$ is the distribution that pools types in the intervals $[0,v_1^*],[v_1^*,v_2^*],\dots,[v_{i-1}^*,v^*_i],[v_i^*,v^*_{i+1}],\dots,[v_N^*,1]$, has support given by
     \begin{equation*}
        \supp G^*_{BHM}= \{w^*_0,w^*_1,w^*_2,\cdots w^*_N\}
    \end{equation*}
and is given by 

\[
G^*_{BHM}(w) = \begin{cases}
0 & \text{if}~w \in [0, w_0^*) \\
F_0(v_1^*) & \text{if}~ w\in [w_0^*,w_1^*]\\
F_0(v_2^*)& \text{if}~ w\in [w_1^*,w_2^*]\\
\dots\\
1 & \text{if}~ w\in[w_N^*,1]
\end{cases}
\]

Note that the way we wrote the solution, we assumed that there is exclusion of the lowest posterior type. This does not necessarily have to be the case, but the result will not change if the lowest posterior type is served. Define $\hat{b}:=\max(w_1^*-v_1^*,w_2^*-v_2^*,\dots,w_N^*-v_N^*)$.

\begin{proposition}\label{low-b}\label{small_b}
    Suppose $b< \hat{b}$. Then, $(q^*_{BHM}, G^*_{BHM})$ does not satisfy the (I-OB) constraint and, thus, cannot be monopolist-optimal.
\end{proposition}

Proposition \ref{small_b} establishes that when the weight the intermediary places on the product quality is small, the seller cannot hope to achieve the same profit as if she were providing information to the buyer directly. Thus, the presence of such an intermediary \textit{hurts} the seller. This result motivates our main question, which is, what is then the optimal menu the seller can post?

\section{Low Bias}
In this section, our analysis focuses on the monopolist's profit-maximizing finite-item menu. We restrict our attention to finite-item menus for two reasons.\footnote{We note that while we have chosen the consumer's value $\theta$ to take values in a continuum, we could alternatively consider the case where there are $K$ consumer types $(\theta_1,\dots,\theta_K)$ with some prior distribution $F_0$. Then, the posterior type $w$ would take values in the interval $[\theta_1,\theta_K]$. In this case, by Winkler (1988), it would be without loss to focus on distributions $G$ over posterior means that have a support of at most $K$ atoms. As a result, the optimal menu would include at most $K$ items. Thus, our results go through with this alternative specification.} First, in real-world scenarios, finite-item menus are exclusively observed. This can be attributed to the existence of fixed costs associated with designing, producing, or marketing products of varying qualities. Moreover, there's the potential burden of escalating costs, especially in marketing, when offering a broad assortment of products. Should these costs prove significant, a monopolist is likely to introduce only those products that can counterbalance them, thereby inherently limiting the diversity of products available in the market (\cite{spence_1980} and \cite{dixit_stiglitz_1977}).

Second, from a theoretical perspective, recent work has established the optimality of finite-item menus in a setting similar to ours. In particular, \cite{bergemann_et_al_2022} study the problem of a monopolist who designs not only the posted menu but also the information structure whose realization the buyer observes. Formally, they study a relaxed version of our problem where there is no intermediary and, thus, no (I-OB) constraint. Their main results show that the solution to the monopolist's problem, in that case, is a finite-item menu together with a monotone pooling distribution. Moreover, they show that when the marginal cost is convex and the prior distribution $F_0$ has either increasing or concave density, the monopolist's optimal mechanism is a single-item menu. 

One natural question is whether the finiteness of the optimal menu continues to hold in our setup.  At an intuitive level, the introduction of an intermediary shouldn't fundamentally change the forces that govern the seller's behavior. That is, pooling qualities and types should remain beneficial for the monopolist. Nevertheless, as we will see, the intermediary's presence does disrupt the seller's pooling ability. This disruption introduces the possibility that the optimal constrained finite menu may under-perform in contrast to a menu combining standalone items with a continuum. The nuanced difference between our study and that of \cite{bergemann_et_al_2022} lies in the intermediary's potential to induce a different distribution over posterior means than the monopolist's optimal when qualities are pooled. While \cite{bergemann_et_al_2022} can analyze the efficiency-rent trade-off within a fixed interval without worrying about what happens outside this interval, we cannot do so in our model. In our scenario, pooling the qualities meant for the types within an interval prompts the intermediary to choose a distribution over posterior means that influences revenues and costs beyond the designated interval.

From now on, we assume that $b<\hat{b}$. Suppose that the seller's optimal menu is an N-item menu. Let $w_1=t_1/q_1$, $w_i=(t_i-t_{i-1})/(q_i-q_{i-1})$ for $i=2,...,N$. The buyer's IC constraints imply that a buyer with type $w<w_1$ does not participate, while if $w_i \leq w<w_{i+1}$, he chooses the item with quality $q_i$ and pays a price of $t_i$. Thus, the intermediary's payoff is given by

\[
u^I(w) = \begin{cases}
0 & \text{if}~w \in \left[0, \frac{t_1}{q_1}\right) \\
\cdots\\
(w+b)q_i-t_i & \text{if}~ w\in \left[\frac{t1}{q1},\frac{t_2-t_1}{q_2-q_1}\right)\\
\cdots \\
(w+b)q_N-t_N & \text{if}~ w\in \left[\frac{t_N-t_{n-1}}{q_N-q_{N-1}}, 1\right]
\end{cases},
\]

It follows that $u^I(w)$ is a piece-wise linear function with jump discontinuities at points $t_1/q_1$, $\{(t_i-t_{i-1})/(q_i-q_{i-1})\}$, $i=2,\dots,N$.  In principle, the intermediary's problem is an intractable Bayesian Persuasion problem to solve. In general, it is not even possible to argue that the optimal distribution will be monotone partitional, as the intermediary's indirect utility $u^I(w)$ does not satisfy the affine-closure property of \cite{d_martini_2019}, which would yield this result. The main methodological contribution of this paper is the following proposition which illustrates that the joint optimality required for $G^*$, that is, the fact that $G^*$ must solve both the seller's and the intermediary's problems is enough to guarantee that $G^*$ will have the form of monotone pooling. Moreover, the types in the support of $G^*$ are completely pinned down by the (I-OB) constraint. This, in turn, reduces the monopolist's problem to the choice of optimal qualities a la \cite{maskin_riley_1984} with the additional caveat that the number of types that participate in the mechanism and, thus, the number of items in the optimal finite-item menu is chosen by the monopolist.

\begin{proposition}\label{opt_types}
    Suppose that an N-item menu $(q^*_i,t^*_i)_{i=1}^N$, together with the distribution $G^*$ constitute the monopolist's optimal outcome among all finite-item menus. Then, $G^*$ pools types in the intervals 
    \begin{align*}
    [0,w_1^*-b],[w_1^*-b,w_2^*-b],...,[w_{i-1}^*-b,w^*_i-b],[w_i^*-b,w^*_{i+1}-b],...,[w_N^*-b,1]
    \end{align*}
    and has support given by
    
    \begin{equation*}
        \supp G^*= \{w^*_0,w^*_1,w^*_2,\cdots w^*_N\}
    \end{equation*}
    where 
\begin{align*}
    w^*_0= & \EE(\theta|0\leq \theta\leq w_1^*-b)\\
    w^*_1 & = t^*_1/q^*_1=\EE(\theta|w_1^*-b\leq\theta\leq w_2^*-b)\\
    w^*_i & =(t^*_i-t^*_{i-1})/(q^*_i-q^*_{i-1})=\EE(\theta|w_i^*-b\leq \theta\leq w_{i+1}^*-b) \ \ \text{for} \ \ i=2,...,N-1\\
    w_N^* & = \frac{t_N^*-t^*_{N-1}}{q_N^*-q_{N-1}^*}=\EE(\theta|w_N^*-b\leq \theta\leq 1)
    \end{align*}

 Moreover, the optimal qualities are given by    
 \begin{equation*}
    c'(q^*_N)=w_N^*  \tag{$\mathrm{OPT-Q_N}$}\\
\end{equation*}
\begin{equation*}
    c'(q^*_i)=w_i^*-\frac{\sum_{j=i}^{N-1} \left(F_0(w^*_{j+1}-b)-F_0(w^*_j-b)\right) }{F_0(w^*_{i+1}-b)-F_0(w^*_i-b)}(w^*_{i+1}-w^*_{i}) \ \ \text{for} \ \ i=1,\dots, N-1 \tag{$\mathrm{OPT-Q_i}$}
    \end{equation*}
\end{proposition}

Proposition \ref{opt_types} shows that the posterior types of the buyer at the optimal finite-item menu are pinned down by the (I-OB) constraint. One can solve
\begin{equation*}
     w_N^*=\EE(\theta|w_N^*-b\leq \theta\leq 1)
\end{equation*}
to obtain $w_N^*$ as a function of $b$. Given $w_N^*(b)$, solve
\begin{equation*}
    w_{N-1}^*=\EE(\theta|w_{N-1}^*-b\leq \theta\leq w_N^*)
\end{equation*}
to obtain $w_{N-1}^*(b)$ and then proceed analogously to obtain all posterior types in the support of $G^*$ as functions of $b$. Then, finding the corresponding qualities reduces to a problem a la \cite{maskin_riley_1984} with the caveat that the number of types in the support of $G^*$ and, therefore, the number of items offered in the menu is determined by the monopolist. Given $w_1^*(b),...,w_N^*(b)$ the monopolist's problem now can be written as 
\begin{align*}
    \max_{(q_1,..., q_N)} & [F_0(w_2^*(b)-b)-F_0(w_1^*(b)-b)][t_1-c(q_1)] \\
   & + \sum_{i=2}^{N-1} [F_0(w_{i+1}^*(b)-b)-F_0(w_{i}^*(b)-b)][t_i-c(q_i)]+[1-F_0(w^*_N(b)-b)][t_N-c(q_n)]
\end{align*}
where $t_1=w_1^*(b)q_1$ and $t_i=t_{i-1}+w_i^*(b)(q_i^*-q_{i-1}^*)$, for $i=2,..,N$. 

Finally, one could simply calculate the profit from different $N$s and find the optimal number of items in the menu. However, it turns out that the optimal number of items is the highest $N$ such that the quality specified by (OPT-Q\textsubscript{i}) is positive.

\begin{proposition}
    The optimal number of items in the menu, $N^*_b$, is the highest possible number such that the lowest quality offered, obtained as the solution to (OPT-Q\textsubscript{1}), is positive. Moreover, $N_b^*$ increases as the intermediary's bias $b$ decreases.  
\end{proposition}

To see this, fix $b$ and suppose that the optimal finite-item menu consists of $N^*_b$ items. What happens if $b$ changes? Notice that Proposition \ref{opt_types} implies that as $b$ decreases, all types $\{w_i^*\}_{i=0,..,N}$ increase. As a result, they all get higher rents than necessary. The seller, instead of offering the optimal $N_b$-item menu, can do better by adding items to the menu. As $b$ decreases, if the seller adds an extra item, there will be enough demand, due to the refined consumer learning that will take place, to support the inclusion of this lower-quality product. This allows the monopolist to fully extract the surplus of the lowest type and make every higher type just indifferent between the quality prescribed for them and the next lower one. However, to achieve this while respecting the optimal behavior of the intermediary, this lower new quality must be offered when $b$ is sufficiently low. Otherwise, the intermediary will steer some consumers to this new product instead of them purchasing higher-quality options, in an undesirable way for the seller.. This would yield a higher rent for them while reducing the measure of posterior types who choose the higher-quality options. As a result, the monopolist's profit would decrease.  The exact value of $b$ below which a new item should be introduced is the one that makes the lowest quality offered in the menu exactly equal to zero. 

This expansion of product offerings will mitigate the profit loss due to the decrease in $b$. However, this loss cannot be reversed. In particular, as $b$ decreases, the monopolist's profit from the best finite-item menu monotonically decreases. At the same time, the buyer's rents and the total welfare in the economy change non-monotonically as a result of the addition of new products to the menu, the price adjustments, and the fact that the probability with which each posterior type happens depends on the bias $b$. A clearer illustration of these facts will be provided in the next section, where we will analyze the simpler and fully tractable Uniform-Quadratic framework.

\section{The Uniform-Quadratic Environment}

We now illustrate the logic of our main results in the Uniform-Quadratic framework. We assume that $b<1/3$ since, otherwise, the solution to the monopolist's problem is the solution to the relaxed problem where the information is provided to the consumer directly by the seller. From Proposition \ref{opt_types}, one can easily see that if an $N$-item menu is the optimal one, it must be the case that 
\begin{align*}
    w_N^*=1-b\\
    w_i^*=w_{i+1}^*-2b=1-(1+2(N-i))b \ \ \text{for} \ \ i=1,\dots N-1 
\end{align*}
Moreover, the probability that each posterior type that trades\footnote{There is also a posterior type $w_0$ who is excluded from the menu, so we ignore this type.} occurs is equal for all such posterior types and given by 
\begin{equation*}
   \prob_{G^*}(w=w_i)=2b \ \ \text{for} \ \ i=1,\dots, N
\end{equation*}
The corresponding qualities are then given by
\begin{align*}
    q_N^*=1-b\\
    q_i^*=w_{i}^*-2(N-i)b=1-(1+4(N-i))b \ \ \text{for} \ \ i=1,\dots N-1 
\end{align*}

For each $b$, the optimal number of items is given by the integer $N_b$ that satisfies 
\begin{equation*}
N_b = \left\lfloor \frac{1}{4} \left( \frac{1}{b} - 1 \right) \right\rfloor + 1
\end{equation*}
where $\left\lfloor \cdot \right\rfloor $  denotes the floor function, which returns the greatest integer less than or equal to its argument.
This means, for instance, that if $b=0.1$, the optimal finite-item menu features 3 items, while if $b=0.01$ it features 25 items and if $b=0.001$ it has 250 items. We now proceed to shed light on the forces that shape how the optimal finite item menu looks like.

\subsection{Optimal Single-Item Menu and Product Variety Expansion}
Suppose the seller posts a single item menu $(q,t)$. Then, the intermediary pools types in $[0,\max(\frac{t}{q}-b,\hat{v})]$ and $ [v,1]$, where $v$ solves 
\begin{equation*}
    \EE(\theta| v \leq \theta\leq 1)=\frac{t}{q}\Leftrightarrow v=2\frac{t}{q}-1
\end{equation*}
Let $w=t/q$. The monopolist's problem is then given by
\begin{equation*}
    \max_{(q,~w)} \left(1-\max(w-b,v)\right)\left(wq-\frac{1}{2}q^2\right)
\end{equation*}

Suppose $(q,p)$ are chosen so that $w-b>v$. The first order conditions yield that, at the optimum,  $w^*=q^*=2(b+1)/3$. However, then we have that $w-b>v$ only if $b\leq 1/5$. Now, suppose that $(t,q)$ are chosen so that $w-b<v$. Then, the first order conditions yield that $w^*=q^*=2/3$ and the constraint is satisfied only if $b\geq 1/3$. For $1/5\leq b<1/3$ the (I-OB) constraint binds and it must be the case that $w^*=q^*=1-b$. Thus, the optimal single-item menu is given by 

\[
w^*=q^* = \begin{cases}
\frac{2}{3} & \text{if}~b \geq \frac{1}{3} \\
1-b & \text{if}~ b\in [\frac{1}{5}, \frac{1}{3})\\
\frac{2(b+1)}{3} & \text{if}~ b\in [0, \frac{1}{5})
\end{cases},
\]

\begin{figure}
\begin{subfigure}{.5\textwidth}
  \centering
  \includegraphics[scale=0.15]{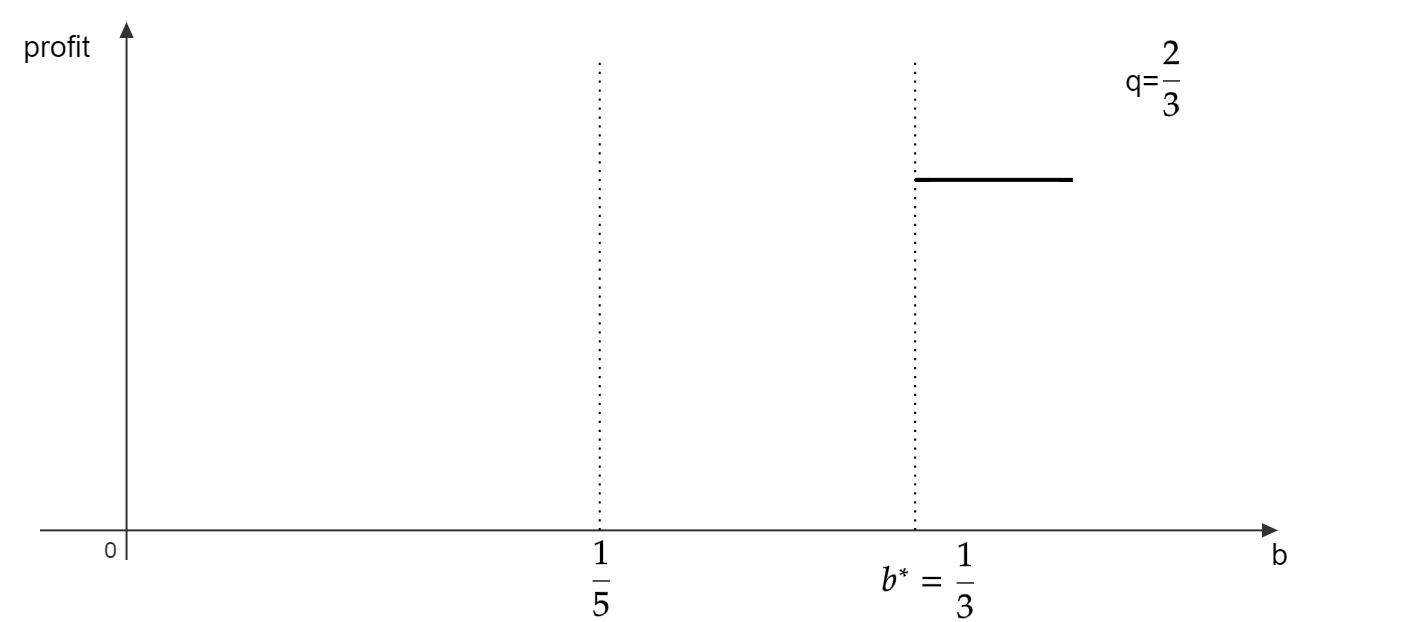}
  \caption{Menu $q_{BHM}=2/3$ is optimal for $b\geq b^*$.}

\end{subfigure}%
\begin{subfigure}{.5\textwidth}
  \centering
  \includegraphics[scale=0.18]{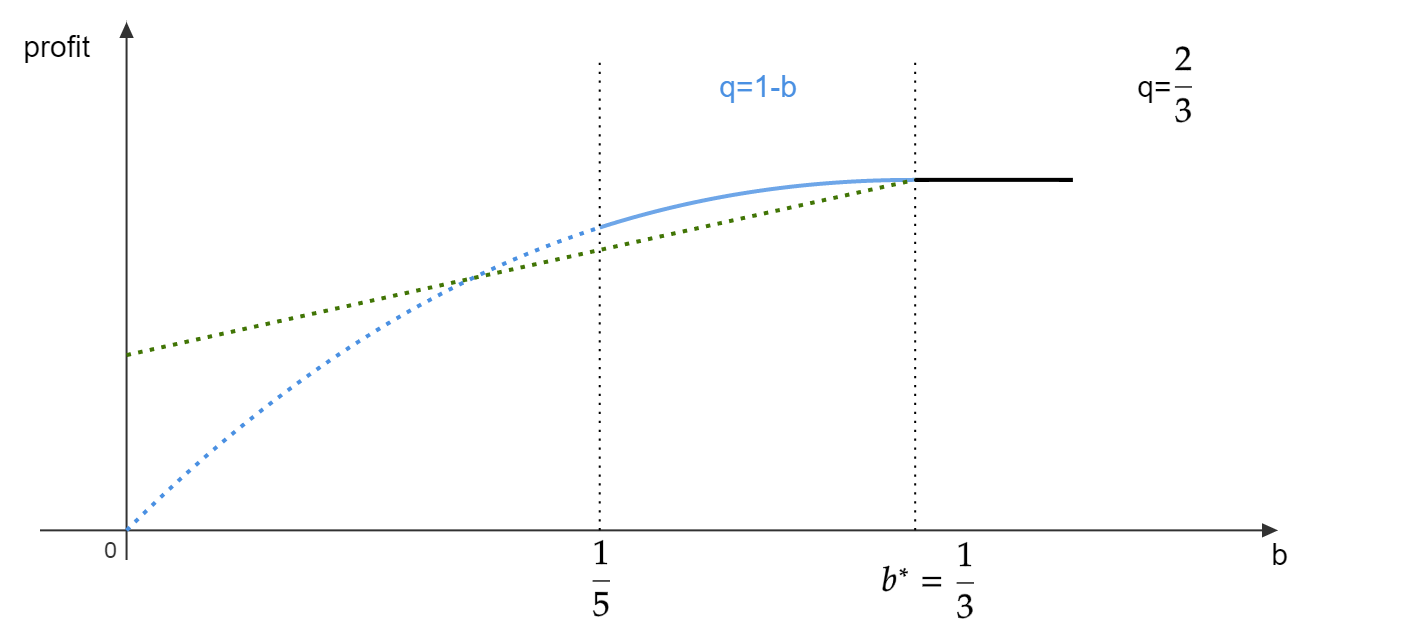}
  \caption{For $b\in[1/5, 1/3)$, $q=1-b$ is optimal.}
 
\end{subfigure}
\begin{subfigure}{.5\textwidth}
  \centering
  \includegraphics[scale=0.18]{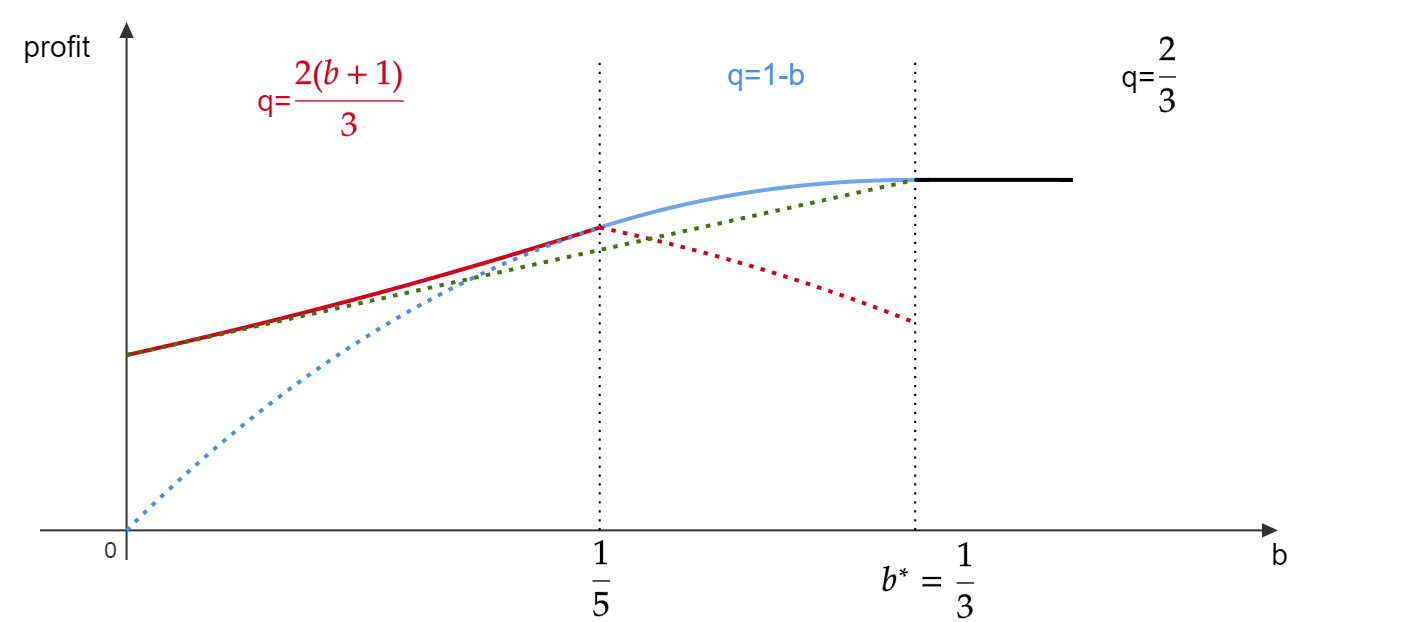}
  \caption{For $b<1/5 $, $q=2(b+1)/3$ is optimal.}
  
\end{subfigure}
\begin{subfigure}{.5\textwidth}
  \centering
  \includegraphics[scale=0.18]{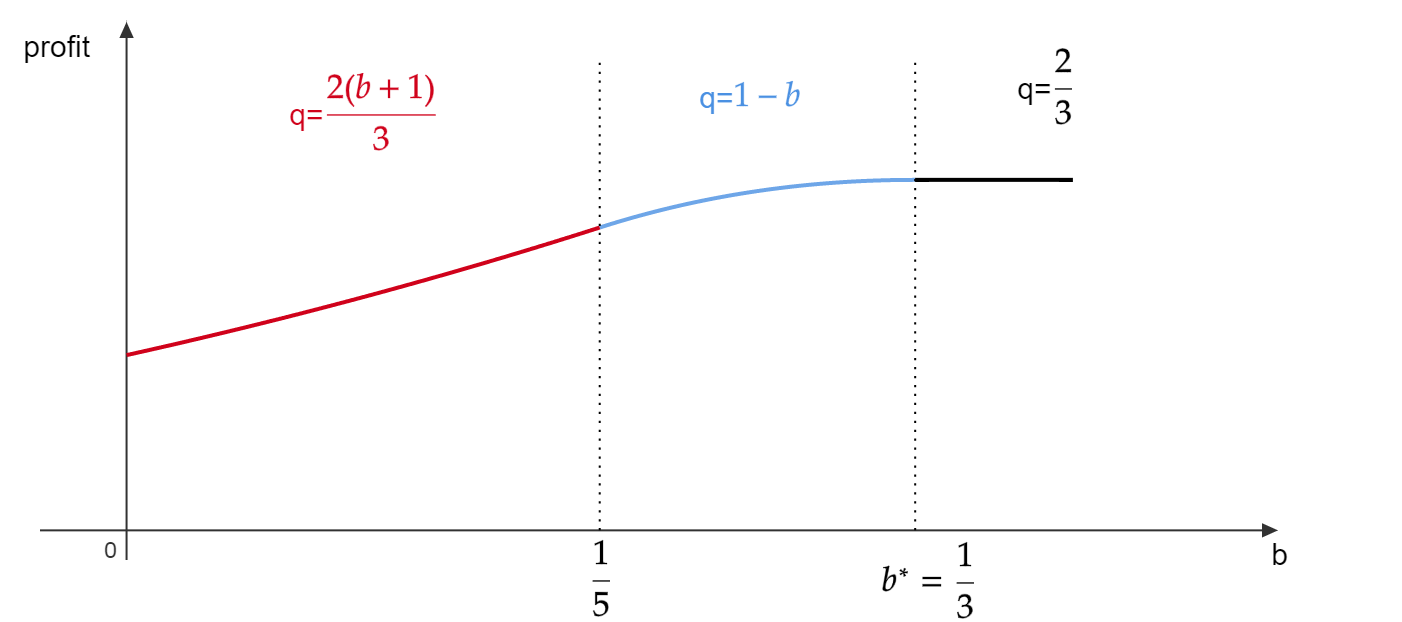}
  \caption{The Optimal Single-Item Menu}
 
\end{subfigure}
\caption{Optimal Single-Item Menu}
\label{opt_1_menu}
\end{figure}

The construction of the optimal single-item menu is illustrated in Figure \ref{opt_1_menu}. Panel (a) shows that when $b\geq b^*$ the menu $q_{BHM}$ is optimal. For $b\in[1/5,1/3)$ the (I-OB) constraint binds and determines the posterior type of the buyer, and, hence, the optimal quality, as illustrated in panel (b). For $b<1/5$ the optimal quality obtained by solving the FOCs to the monopolist's problem satisfies the (I-OB) constraint and so it is optimal, as illustrated in panel (c). 

A similar approach yields the optimal 2-item menu. In particular, we have that it reduces to the optimal single-item menu for $b\geq 1/5$. It includes two quality options $q_1=1-5b$ and $q_2=1-b$ for $1/9 \leq b<1/5$ and for $b<1/9$ it includes $q_1=2(b+1)/5$ and $q_2=4(b+1)/5$. The optimal 2-item menu is illustrated in Figure 5 (purple solid and dashed line).

\begin{figure}

  \centering
  \includegraphics[scale=0.3]{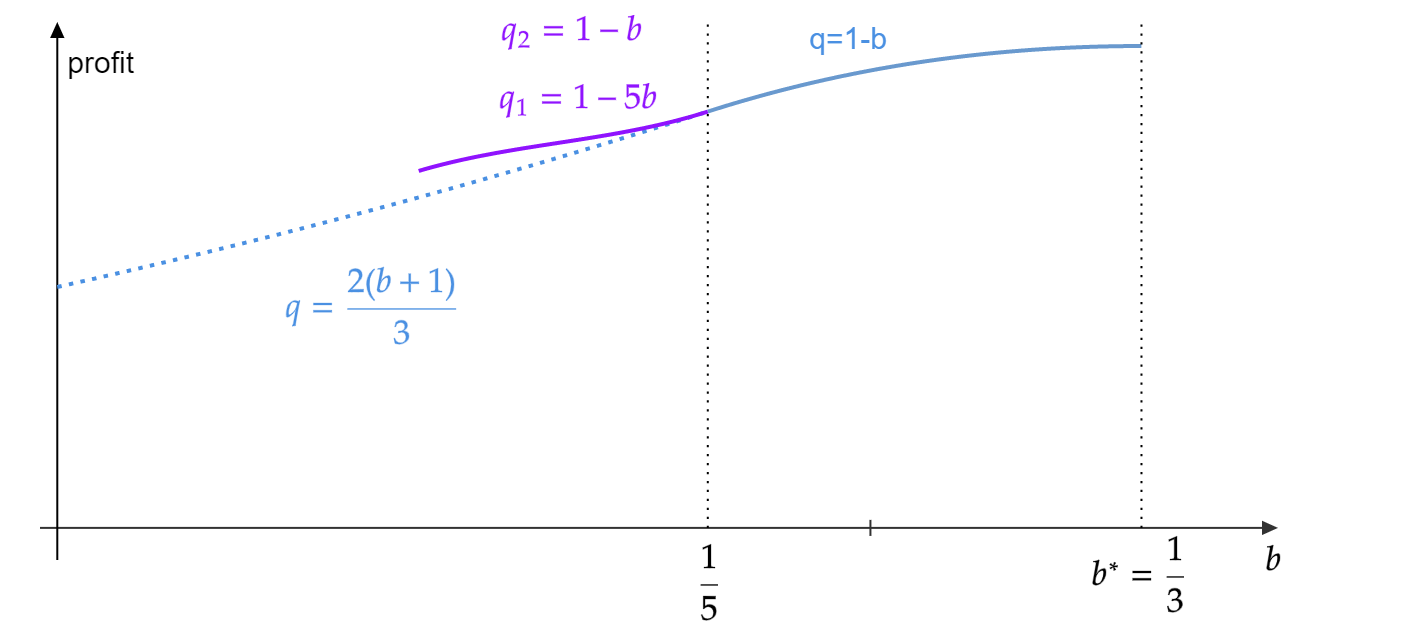}
  \caption{Introducing a lower-quality option and increasing the existing one yields a higher profit to the monopolist.}

\end{figure}

However, the monopolist can do better by introducing lower-quality items in the menu as $b$ decreases. In particular, for $b<1/5$, instead of offering the single quality $q=2(b+1)/3$ and leaving rents to the posterior buyer type, she can increase this quality and introduce a lower-quality option, that is, offer two items. The low posterior type of the buyer will get zero surplus under the new menu and the high posterior type is just indifferent between the two options. This is illustrated in Figure 4. The monopolist must keep doing that as $b$ decreases to mitigate the profit loss. In particular, introducing a lower-quality product is optimal whenever the marginal cost of producing the lowest quality prescribed by Proposition \ref{opt_types} is positive.

\begin{figure}

  \centering
  \includegraphics[scale=0.3]{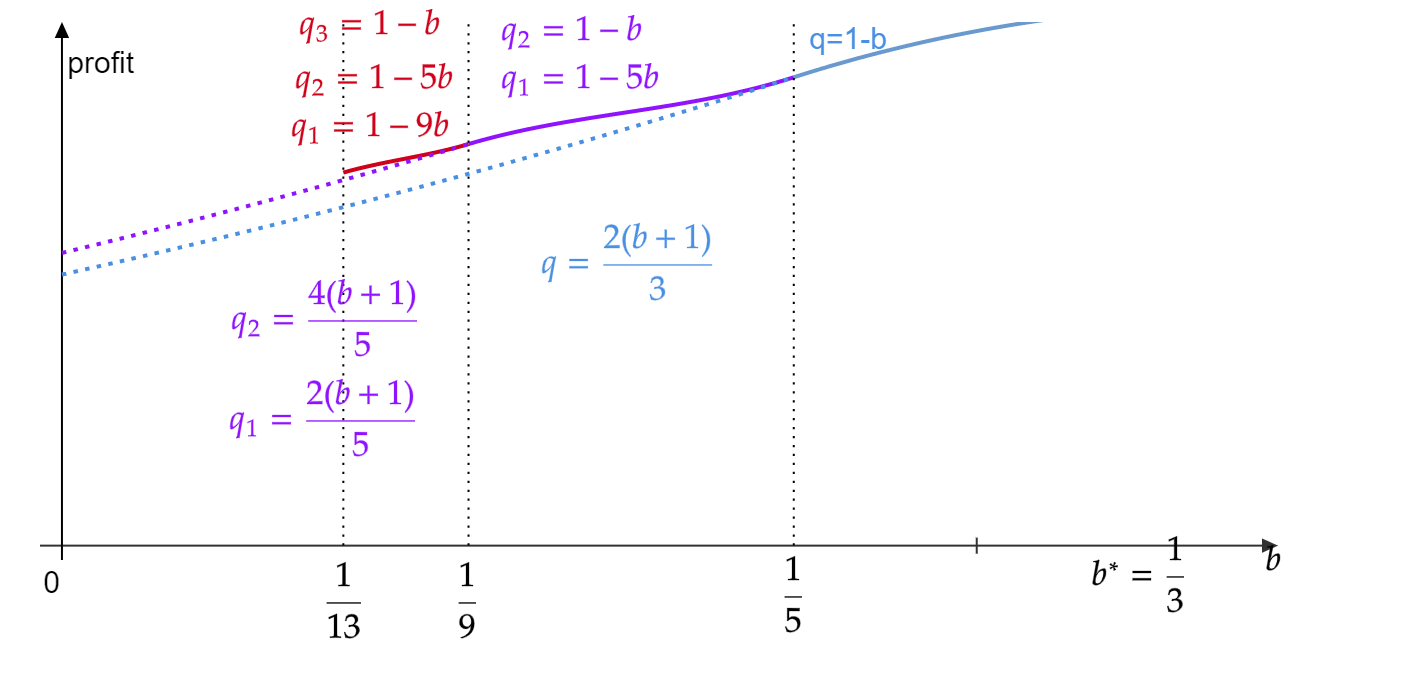}
  \caption{The monopolist must keep adding options to mitigate the profit loss as $b$ decreases.}

\end{figure}

\subsection{Comparative Statics:  Total Surplus, Profits, Rents, and Distortions}
We now proceed to further explore the effects of an intermediary's presence on market outcomes of interest. To facilitate a clear and comprehensive analysis, we will maintain our focus on the Uniform-Quadratic environment leveraging its full tractability.

Our argument hinges on the observation that the intermediary’s presence necessitates adjustments from the monopolist, specifically in the form of augmenting the optimal menu to counteract potential profit losses. This strategic response has important implications for the overall functioning of the market. In particular, the consumer's rents, the total surplus, and the average distortions in the market are non-monotone in the bias term. This leads to a surprising observation; on one hand, consumers would exhibit, if they had the option, a clear preference for intermediaries with minimal bias, as this configuration aligns more closely with their interests. On the other hand, when the disparity in bias between two intermediaries is marginal, a scenario may arise where an intermediary with a slightly higher level of bias could be deemed preferable. This interplay between intermediary bias and consumer surplus underscores the need for an understanding of how intermediary behavior influences market mechanisms and consumer welfare, especially when one has the goal of designing policies that facilitate transparency in the market.

Note that the total surplus, rents and distortions are given respectively by:
\begin{align*}
    TS=\EE_{G^*}[wq(w)-c(q(w))] \tag{Total Surplus}\\
    R=\EE_{G^*}[wq(w)-t(w)] \tag{Rents}\\
     D=\EE_{G^*}[q^{FB}(w)-q(w)] \tag{Distortons}
\end{align*}

Where $q^{FB}(w)$ is the efficient quality that type $w$ should receive. Since in our framework we have $c'(q)=q$, it follows that $q^{FB}(w)=w$.
The next three figures illustrate how these functions behave for different values of the intermediary's bias $b$. The different colors capture the different finite-item menus that are optimal as $b$ decreases.

Some comparisons are in order. When juxtaposed with a situation where the seller is the information provider, several differences become apparent. In our intermediary-driven context, while the consumer rents increase, there is a decrease in the average quality of products presented on the menu, as a result of the product range expansion. The intermediary's presence leads to a ``trade-off" of average product quality for enhanced buyer surplus. Additionally, fewer posterior consumer types participate in the mechanism as a result of their superior information regarding their match with products. Those consumers who do engage, receive products of inferior quality on average, yet they achieve higher rents. Consequently, the overall market surplus experiences a reduction.

In contrast, when compared to the standard monopolistic screening scenario---where consumers possess private knowledge of their valuations---the predictions change. In particular, the market exhibits an increase in total surplus as a result of higher monopolist profits despite the decrease in consumer rents. Furthermore, the extent of quality degradation is less pronounced, that is, on average, posterior consumer types receive qualities that are closer to the efficient level.

\begin{figure}
\begin{subfigure}{.5\textwidth}
  \centering
  \includegraphics[scale=0.5]{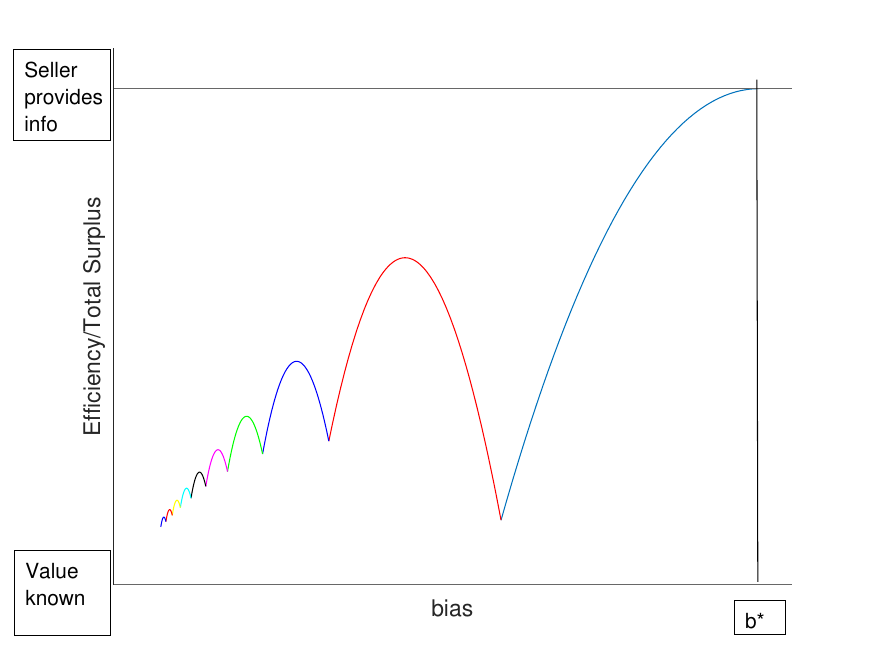}
  \caption{Total Surplus as $b$ changes.}

\end{subfigure}%
\begin{subfigure}{.5\textwidth}
  \centering
  \includegraphics[scale=0.5]{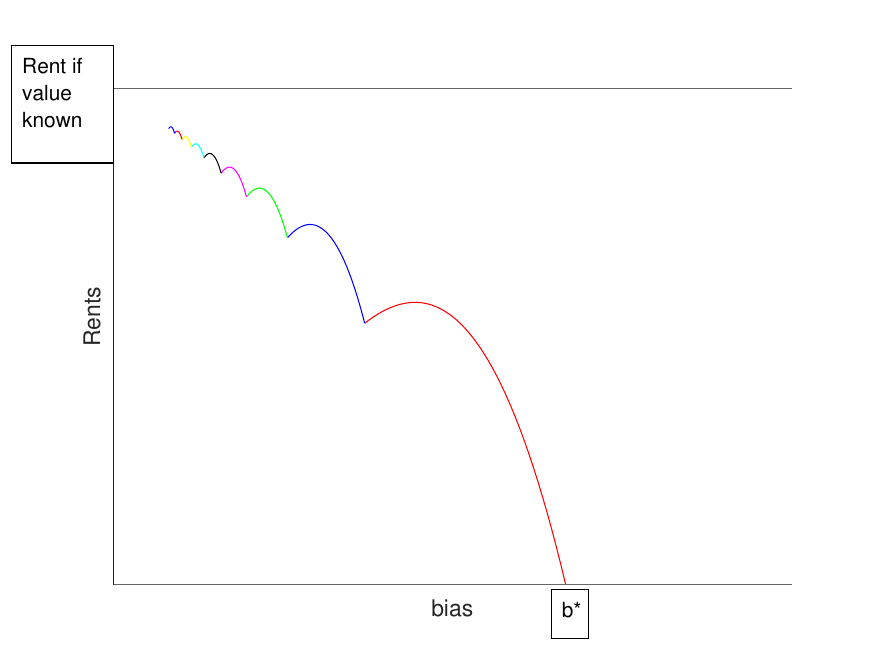}
  \caption{Rents as $b$ changes.}

\end{subfigure}
\caption{Total surplus and rents}
\end{figure}

The significance of these findings is important for applied researchers. Screening models effectively encapsulate the critical features of significant markets, encompassing sectors from cable television and health care to mobile telecommunications. There exists an extensive body of empirical research that delves into these models, with a specific focus on nonlinear pricing; a recent comprehensive survey can be found in \cite{perrigne_vuong}.

Notably, bringing data to these models enables researchers to evaluate quality deterioration and assess the degree to which firms wielding market power spawn inefficiencies within contexts plagued by informational asymmetries. This literature typically uses the first-order conditions of the monopolist's problem to identify variables of interest, thus placing substantial emphasis on the precision of the model specification.  Our results suggest that, by ignoring the presence of intermediaries, this approach may incorrectly estimate the real level of quality deterioration. Suppose that a researcher ignores the intermediary and tries to estimate quality degradation using the screening results of \cite{mussa_rosen_1978}. The researcher observes the distribution of signed contracts and the associated product qualities. Her goal is to estimate the distribution of types in the economy with the purpose of estimating the quality wedge between average posterior buyer types and qualities offered. As can be seen in Figure \ref{distortions}, ignoring the intermediary would result in overestimating the inefficiency in the market. 

\begin{figure}
  \centering
  \includegraphics[scale=0.5]{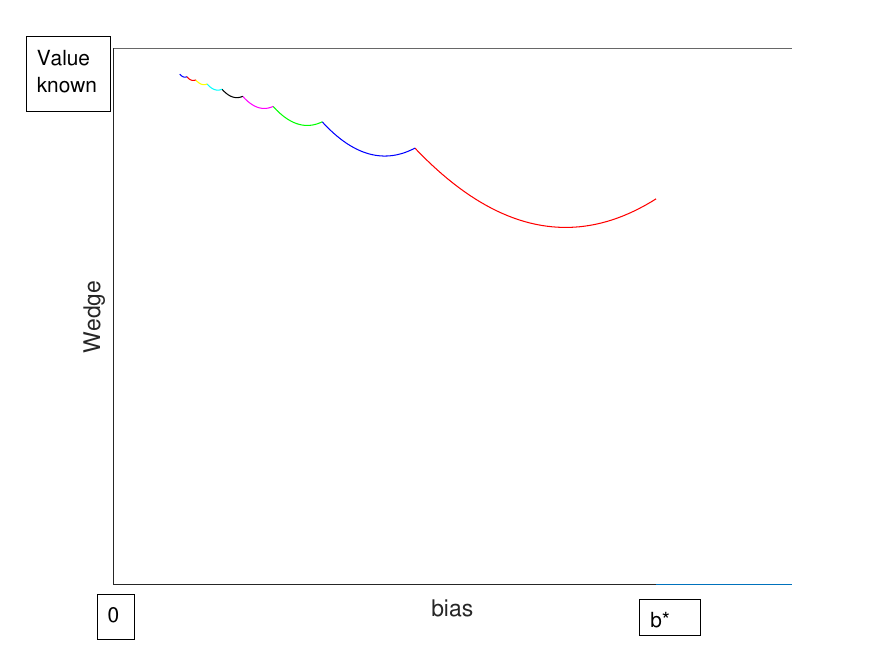}
  \caption{Distortions as $b$ changes.}
  \label{distortions}
  \end{figure}

If, on the other hand, she based her estimation on the theory developed in \cite{bergemann_et_al_2022}, with the purpose of estimating the inefficiencies present when the monopolist provides information to the consumer, she would underestimate such inefficiencies.

\subsection{Exogenous Restrictions on the Number of Items}

The analysis of finite-item menus garners significance due to practical considerations, typically implying limitations on the extent to which sellers can diversify their product offerings.To elucidate, consider the scenario where 
$b=0.001$ within the uniform-quadratic environment; the derived optimal menu features 250 items, a quantity that, realistically, a monopolist may find challenging to offer. This brings us to the crucial analysis of scenarios where the number of items a monopolist can include in the menu is externally constrained.

In situations characterized by such constraints, we observe that consumers tend to receive higher informational rents compared to scenarios where the seller can offer the optimal number of items. Despite this increase in informational rents, consumers find themselves at a disadvantage as their overall payoff is compromised relative to conditions allowing for product expansion. Similarly, the monopolist's profit also takes a hit under this imposed limitation, resulting in a net decrease in total surplus.

This underscores an important insight: the capability of a monopolist to augment the menu with an optimal variety of items stands to benefit both the seller and the consumer. Figures 8 and 9 provide a visual representation, illustrating the advantageous outcomes associated with optimal product expansion, and thus, emphasizing the criticality of such flexibility in product offerings for maximizing overall economic welfare.

Specifically, suppose that (i) the monopolist can offer up to two items on the menu and (ii) the monopolist can offer up to three items on the menu.\footnote{We use at most two or three items as a restriction observable in practice. Of course, the exact upper bound on the number of items is irrelevant.} In the previous section, we derived the optimal two-item menu. One can show that with at most three items, the monopolist-optimal menu is given by:
\begin{figure}
  \centering
  \includegraphics[scale=0.8]{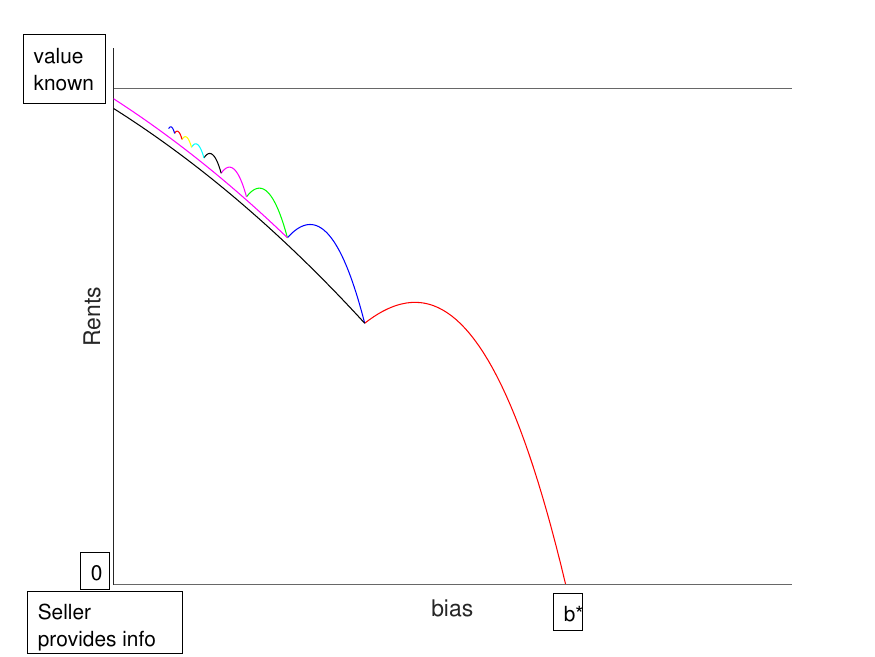}
  \caption{(i)Red-Black line: rent from at most 2 items.
   (ii) Red-Blue-Magenda line: rent from at most 3 items.}
  \label{restriction_on_items_rents}
  \end{figure}

\begin{itemize}
    \item If $b\geq 1/3$, offer a single item $q_1=\frac{2}{3}$.
    \item If $\frac{1}{5}\leq b <\frac{1}{3}$, offer a single item $q_=1-b$.
    \item If $\frac{1}{9}\leq b<\frac{1}{5}$, offer two items, $q_1=1-5b$ and $q_2=1-b$.
    \item If $\frac{1}{13}\leq b<\frac{1}{9}$, offer three items, $q_1=1-9b$, $q_2=1-5b$ and $q_3=1-b$.
    \item If $b<\frac{1}{13}$, offer three items, $q_1=\frac{2(b+1)}{7}$, $q_2=\frac{4(b+1)}{7}$ and $q_3=\frac{6(b+1)}{7}$.
\end{itemize}

\begin{figure}
  \centering
  \includegraphics[scale=0.8]{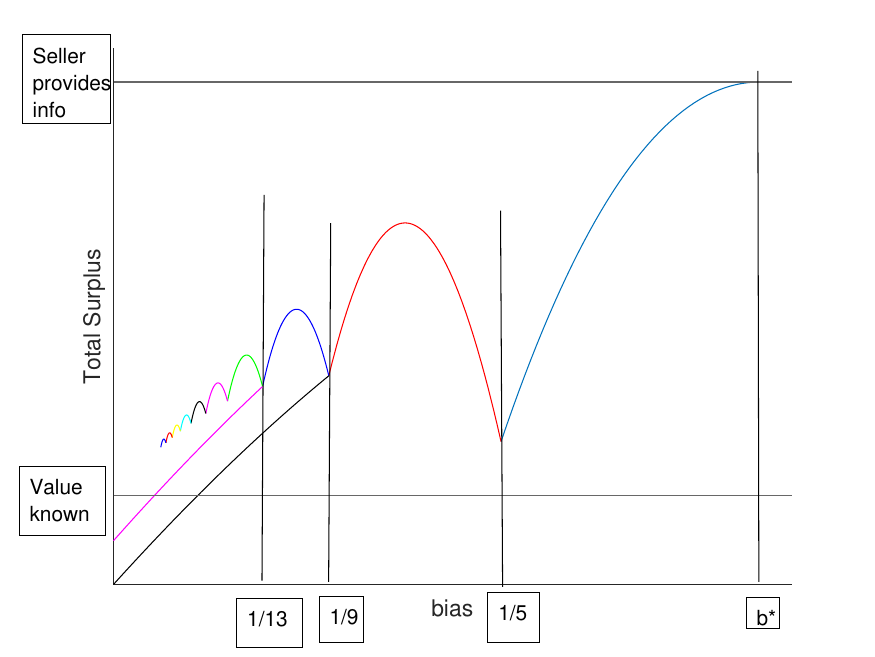}
  \caption{(i) Cyan-Red-Black line: total surplus from at most 2 items. 
  (ii) Cyan-Red-Blue-Magenta line: total surplus from at most 3 items.}
  \label{restriction_on_items_ts}
  \end{figure}

The limitation on the assortment of items becomes particularly pertinent when considering low values of $b$. Specifically, when b falls below $1/9$, the monopolist finds it advantageous to add a third product to her offerings. Similarly, when 
$b$ drops below $1/13$, incorporating a fourth item becomes the strategy. In both cases, the trajectory follows that with a further decrease in the intermediary's bias term, the monopolist would seek to introduce additional products of lesser quality to the menu.

\begin{figure}
  \centering
  \includegraphics[scale=0.8]{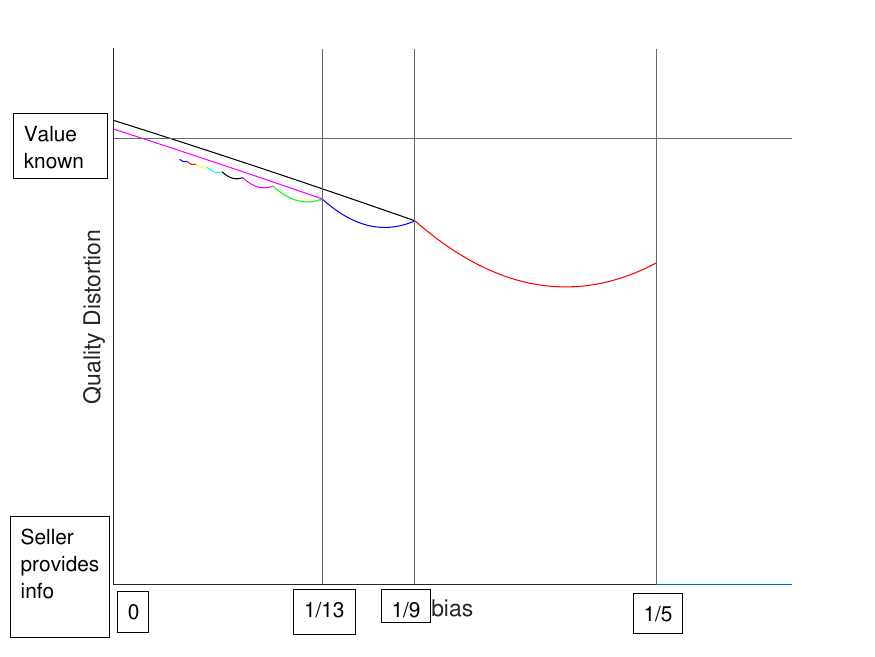}
  \caption{(i) Red-Black line: quality distortions from at most 2 items. 
  (ii) Red-Blue-Magenta line: quality distortions from at most 3 items.}
  \label{restriction_on_items_d}
  \end{figure}

Moreover, as we saw, the quality distortions when an intermediary is present are lower than in the standard screening environment. However, if the maximum number of products on the menu is exogenously restricted, these distortions will be higher than in the standard screening case when the intermediary's bias is sufficiently small, as illustrated in Figure 10.   

\section{Discussion: Connection to the Information Acquisition Framework}

Two recent papers by \cite{thereze_2022} and \cite{mensch_ravid_2022}, study a similar environment to ours with a notable distinction: rather than incorporating an intermediary, these works assume that consumers can obtain information about their product valuation, albeit at a cost. This alternate setup is also clearly natural, as it mirrors the real-world coexistence of both intermediaries and information-seeking consumers. We briefly compare and contrast our approach and results with that of \cite{thereze_2022} and \cite{mensch_ravid_2022}.

A commonality across all three studies is the assumption that consumers initially lack private information regarding the value of various products upon observing the menu posted by the monopolist. Consequently, the seller must anticipate consumer learning that occurs post-revelation of the menu. This requires the seller to strategically design the menu to influence either the intermediary's choices in our model or the consumer's information acquisition in \cite{thereze_2022} and \cite{mensch_ravid_2022}, all to her benefit. Methodologically, this strategic element imposes an extra constraint in the monopolist’s optimization problem in both scenarios. Yet, these constraints differ fundamentally, leading to a primary divergence: in \cite{thereze_2022} and \cite{mensch_2022}, the monopolist is compelled to leave the consumer moral-hazard rents because she cannot contract based on the consumer's learning decisions. This occurs as the consumer can shift the agency relationship's power balance by opting to learn, potentially against the seller's preferences. The rents the consumer receives stem not just from actual private information but also from potential information he could acquire. This implies that the distortions arising from costly information acquisition are more pronounced than what standard monopolistic screening predicts as per \cite{mussa_rosen_1978} and \cite{maskin_riley_1984}. Both \cite{thereze_2022} and \cite{mensch_ravid_2022} show that the no distortion-at-the-top result may be reversed. Contrarily, our model does not mirror this; the optimal finite-item menu features no distortion at the top. The reason is that in our framework, the monopolist can freely increase the quality she provides to the highest posterior type of consumer and the corresponding price, increasing her profits without violating the intermediary's obedience constraint. Additionally, our framework implies lesser overall distortions compared to a scenario where consumers have private information about their valuations.  Hence, researchers using our model might overestimate market inefficiencies, while those employing the models of \cite{thereze_2022} or \cite{mensch_ravid_2022} will underestimate them. Moreover, \cite{thereze_2022} shows that profits and consumer surplus are non-monotonic in the level of information costs, which is the parameter of interest in his model. In our work, the non-monotonicity of consumer payoffs is also present, but profits are decreasing in the intermediary's bias, which is our parameter of interest.

However, our core contribution relative to \cite{thereze_2022} and \cite{mensch_ravid_2022}, apart from analyzing this different framework which is clearly relative to understanding real-world questions, is highlighting the monopolist’s necessity to diversify the product range in the optimal menu as intermediary bias diminishes, thereby refining consumer learning. In essence, the seller must contemplate how her product offerings will impact consumer learning, given that consumer willingness to pay is information-dependent. For example, it might be futile to introduce innovative product attributes if consumers don’t make an effort to learn about them pre-purchase. The presence of an intermediary changes this, as consumers receive this information ``without charge." This insight also explains why sellers might pay commissions to influencers or specialized reviewers. This strategy not only allows some level of control over consumer learning, benefiting the seller but also reduces the necessity to offer an extensive variety of products, which might be impractical for various external reasons.

\section{Conclusion}

In this paper, we study a screening in which the consumer's endogenous private information is provided by an intermediary. While this intermediary wishes to maximize the consumer's payoff, he is also biased toward high-quality products. 

We characterize the profit-maximizing finite-item menu. Specifically, this menu coincides with the Mussa-Rosen menu if the intermediary's and the consumer's preferences coincide. Conversely, when the bias is large, it coincides with the one the seller would offer if she were providing information directly to the buyer.

As the intermediary's bias diminishes, the monopolist's optimal menu involves an expanded variety of products. This expansion mitigates the profit loss due to refined buyer's learning that the smaller intermediary's bias makes possible. Interestingly, consumer surplus exhibits a non-monotonic relationship with the intermediary's bias. While buyers prefer an intermediary with minimal bias over one with extreme bias, there can be scenarios where they favor the highly biased intermediary when the bias difference is not too significant. Furthermore, the seller's ability to respond to more accurate buyer's learning by offering a wider array of products proves beneficial for both the seller and the consumer and, therefore, increases the overall efficiency in the economy, especially relative to the case where the seller is restricted on the number of products she can offer.

Our model serves as a stylized benchmark and leaves several questions unanswered for future research. Firstly, we have assumed no interaction between the intermediary and the seller and have ruled out any transfer schemes that involve the intermediary. In reality, intermediaries often maintain contractual relationships with sellers. Investigating a model where sellers can use transfers to influence intermediary behavior is both interesting and applicable. Second, we have assumed that the information provision by the intermediary to the consumer takes place \textit{after} the seller posts the menu. The reverse timing is plausible as well. In future work, we plan to study how the optimal menu and the intermediary's information provision change in settings in which the intermediary provides information to the consumer before the monopolist posts the menu. Lastly, we intend to consider a broader range of the intermediary's objectives.
\medskip

\clearpage

\bibliographystyle{te}
\bibliography{references}

@unpublished{bergemann_et_al_2022,
    author = {Dirk Bergemann and Tibor Heumann and Stephen Morris},
    title = {Screening with Persuasion},
    note = {{a}rXiv preprint arXiv:2212.03360},
    year = {2022},
}

@article{deniz_kovac_2019,
    author = {Dizdar Deniz and Eugen Kovac},
    title = {A Simple Proof of Strong Duality in the Linear Persuasion Problem},
    journal = {Games and Economic Behavior},
    volume = {122},
    year = {2020},
    pages = {407--412},
}

@article{dixit_stiglitz_1977,
    author = {Avinash K. Dixit and Joseph E. Stiglitz},
    title = {Monopolistic Competition and Optimum Product Diversity},
    journal = {American Economic Review},
    volume = {67},
    year = {1977},
    pages = {297--308},
}

@article{spence_1980,
    author = {Michael A. Spence},
    title = {Multi-product Quantity-Dependent Prices and Profitability Constraints},
    journal = {Review of Economic Studies},
    volume = {47},
    year = {1980},
    pages = {821--841},
}

@article{perrigne_vuong,
    author = {Isabelle Perrigne and Quang Vuong},
    title = {Econometrics of Auctions and Nonlinear Pricing},
    journal = {Annual Review of Economics},
    volume = {11},
    year = {2019},
    pages = {27--54},
}

@article{bryn_hu_smith_2003,
    author = {Erik Brynjolfsson and Yu (Jeffrey) Hu and Michael D. Smith},
    title = {Consumer Surplus in the Digital Economy: Estimating the Value of Increased Product Variety at Online Booksellers},
    journal = {Management Science},
    volume = {49},
    number = {11},
    year = {2003},
    pages = {1580--1596},
}

@article{roessler_szentes_2017,
    author = {Anne-Katrin Roesler and Balázs Szentes},
    title = {Buyer-Optimal Learning and Monopoly Pricing},
    journal = {American Economic Review},
    volume = {107},
    number = {7},
    year = {2017},
    pages = {2072--2080},
}

@article{roessler_ravid_szentes_2022,
    author = {Doron Ravid and Anne-Katrin Roesler and Balázs Szentes},
    title = {Learning before Trading: On the Inefficiency of Ignoring Free Information},
    journal = {Journal of Political Economy},
    volume = {130},
    number = {2},
    year = {2022},
    pages = {346--387},
}

@article{ksm_2021,
    author = {Andreas Kleiner and Benny Moldovanu and Philipp Strack},
    title = {Extreme Points and Majorization: Economic Applications},
    journal = {Econometrica},
    volume = {89},
    number = {4},
    year = {2021},
    pages = {1557-1593},
}

@article{d_martini_2019,
    author = {Piotr Dworczak and Giorgio Martini},
    title = {The Simple Economics of Optimal Persuasion},
    journal = {Journal of Political Economy},
    volume = {127},
    number = {5},
    year = {2019},
    pages = {1993--2048},
}

@article{kolotolin_et_al_2022,
    author = {Anton Kolotilin and Timofiy Mylovanov and Andriy Zapechelnyuk},
    title = {Censorship as Optimal Persuasion},
    journal = {Theoretical economics},
    volume = {17},
    year = {2022},
    pages = {561--585},
}

@article{arieli_et_all_2023,
    author = {Itai Arieli  and Yakov Babichenko and Rann Smorodinsky and Takuro Yamashita},
    title = {Optimal Persuasion via bi-pooling},
    journal = {Theoretical Economicsy},
    volume = {18},
    number = {1},
    year = {2023},
    pages = {15--36},
}

@article{malenko_tsoy_2019,
    author = {Andrey Malenko and Anton Tsoy},
    title = {Selling to Advised Buyers},
    journal = {American Economic Review},
    volume = {109},
    number = {4},
    pages = {1323--1348},
    year = {2019},
}

@article{maskin_riley_1984,
    author = {Eric Maskin and John Riley},
    title = {Monopoly with Incomplete Information},
    journal = {RAND Journal of Economics},
    volume = {15},
    number = {2},
    pages = {171--196},
    year = {1984},
}

@article{mensch_2022,
    author = {Jeffrey Mensch},
    title = {Screening Inattentive Buyers},
    journal = {American Economic Review},
    volume = {112},
    note = {6},
    year = {2022},
    pages = {1949--1984},
}

@unpublished{mensch_ravid_2022,
    author = {Jeffrey Mensch and Doron Ravid},
    title = {Monopoly, Product Quality, and Flexible Learning},
    note = {{arXiv} preprint arXiv:2202.09985},
    year = {2022},
}

@article{mussa_rosen_1978,
    author = {Michael Mussa and Sherwin Rosen},
    title = {Monopoly and Product Quality},
    journal = {Journal of Economic Theory},
    volume = {18},
    number = {2},
    pages = {301--317},
    year = {1978},
}

@unpublished{thereze_2022,
    author = {{Jo\~{a}o} Thereze},
    title = {Screening Costly Information},
    note = {Working Paper, Princeton University},
    year = {2022},
}
\newpage
\medskip

\section*{Appendix}

\subsection*{Proof of Proposition \ref{existence}}

Fix the menu that the monopolist offers. Let $X=[0,\Bar{q}]\times \mathbb{R}_+$ and let $\Delta (X\times[0,1])$ be the set of Borel measures over $X\times [0,1]$. We endow $\Delta (X\times[0,1])$ with the weak* topology. Given the menu the monopolist offers, the intermediary's problem is a standard linear persuasion problem. Assuming that the buyer breaks indifferences in favor of the intermediary, the intermediary's utility function $u^I(\cdot)$ is upper-semi-continuous. The constraint set, that is, the set of mean preserving contractions of the prior, $MPC(F_0)$ is compact, and, thus, the intermediary's problem has a solution set $K(M)$ which is non-empty for every menu $M$. Let $\Tilde{X}$ be the collection of compact subsets of $X$ that contain the non-participation option $(0,0)$. Also, notice that we can assume that $\Bar{X}=[0,\Bar{q}]\times[0,\Bar{q}]$ since the buyer strictly prefers non-participation that paying transfer higher than $\Bar{q}$. The monopolist's problem is then written as
\begin{align*}
\max_{(M,G)\in\Tilde{X}\times\Delta (X\times[0,1])} \int t-c(q) G(d(q,t,\theta)) \\
\text{s.t.} \ \ G\in K(M)
\end{align*}
Let $\mathbf{\Tilde{X}}$ be the set of all compact non-empty subsets of $\Bar{X}$ and endow it with the Hausdorff metric. Denote by $\Bar{K}$ the restriction of $K(M)$ to $M\in\mathbf{\Tilde{X}}$. Then, let $GR (\Bar{K})$ be the graph of the restriction:
\begin{equation*}
    GR (\Bar{K})=\{(M,G)\in\mathbf{\Tilde{X}}\times\Delta (\Bar{X}\times[0,1]): G\in \Bar{K}(M)  \}
\end{equation*}
We can re-write the monopolist's problem as

\begin{align*}
\max_{(M,G)\in GR(\Bar{K})} \int t-c(q) G(d(q,t,\theta)) 
\end{align*}

Now, since $\mathbf{\Tilde{X}}$ is compact, it follows by Berge's maximum theorem that the correspondence $\Bar{K}$ is upper-hemi-continuous and has a closed graph. It follows that $GR (\Bar{K})$ is compact since it is a subset of the compact set $\mathbf{\Tilde{X}}\times\Delta (\Bar{X}\times[0,1])$. Thus, the monopolist's problem admits a solution. \qed

\subsection*{Proof of Proposition 2}

The proof follows from the arguments in the main text preceding the Proposition. \qed

\subsection*{Proof of Proposition 3}
Under Assumption 1 and $c'''(q)\ge 0$, the solution to the relaxed problem is a single-item
menu $(q^*,t^*)$ together with a cutoff $v^*$ and posterior mean
\[
w^*:=E(\theta\mid v^*\le \theta\le 1),
\]
where $t^*=w^*q^*$ and $(v^*,q^*)$ solve (OSIM). Let
\[
w_0^*:=E(\theta\mid 0\le \theta\le v^*).
\]

If the seller posts the single-item menu $(q^*,t^*)$, then the buyer purchases if and only if
\[
wq^*-t^*\ge 0
\quad\Longleftrightarrow\quad
w\ge \frac{t^*}{q^*}=w^*.
\]
Hence the intermediary's indirect utility is
\[
u_I(w)=
\begin{cases}
0,& w\in[0,w^*),\\
(w+b)q^*-t^*,& w\in[w^*,1].
\end{cases}
\]

The distribution induced by pooling states in $[0,v^*]$ and $[v^*,1]$ is the two-atom
distribution
\[
G_{BHM}^*(w)=
\begin{cases}
0,& w<w_0^*,\\
F_0(v^*),& w_0^*\le w<w^*,\\
1,& w\ge w^*.
\end{cases}
\]

We claim that $G_{BHM}^*$ solves the intermediary's problem if and only if
\[
b\ge w^*-v^*.
\]

First suppose that $b\ge w^*-v^*$. Define
\[
\lambda:=\frac{bq^*}{w^*-v^*}.
\]
Then $\lambda\ge q^*$. Consider
\[
p^*(w)=
\begin{cases}
0,& w\in[0,v^*],\\
\lambda (w-v^*),& w\in[v^*,1].
\end{cases}
\]
We verify that $p^*$ is a price function for $G_{BHM}^*$.

(i) $p^*\ge u_I$:
for $w<w^*$, we have $u_I(w)=0$, while $p^*(w)\ge 0$.
For $w\ge w^*$,
\[
p^*(w)-u_I(w)
=
\lambda (w-v^*)-q^*(w+b-w^*)
=
(\lambda-q^*)(w-w^*)\ge 0,
\]
where we used $\lambda(w^*-v^*)=bq^*$.

(ii) $p^*$ is convex, since its slope is $0$ on $[0,v^*]$ and $\lambda$ on $[v^*,1]$.

(iii) $p^*=u_I$ on $\operatorname{supp} G_{BHM}^*$:
because $w_0^*\le v^*$, we have $p^*(w_0^*)=0=u_I(w_0^*)$, and
\[
p^*(w^*)=\lambda (w^*-v^*)=bq^*=u_I(w^*).
\]

(iv) $\int p^*\,dG_{BHM}^*=\int p^*\,dF_0$:
since $p^*=0$ on $[0,v^*]$ and affine on $[v^*,1]$,
\[
\int p^*(w)\,dF_0(w)
=
(1-F_0(v^*))\,p^*(w^*)
=
\int p^*(w)\,dG_{BHM}^*(w).
\]

Thus $p^*$ is a price function for $G_{BHM}^*$, so $G_{BHM}^*$ solves the intermediary's
problem. Therefore $(q_{BHM}^*,G_{BHM}^*)$ satisfies (I-OB). Since it is optimal in the
relaxed problem, it is also optimal in the original problem.

Now suppose that $b<w^*-v^*$ and, toward a contradiction, that $G_{BHM}^*$ solves the
intermediary's problem. Since $u_I$ is piecewise affine with a single jump, it is Lipschitz on
a neighborhood of $0$ and on a neighborhood of $1$. Hence, by Deniz and Kovac (2020),
there exists a price function $p$ for $G_{BHM}^*$.

Because $G_{BHM}^*$ pools the interval $[0,v^*]$ at its mean $w_0^*$ and the interval
$[v^*,1]$ at its mean $w^*$, Jensen's inequality and condition (iv) imply that $p$ must be
affine on each of the two intervals $[0,v^*]$ and $[v^*,1]$.

On $[0,v^*]$, we have $u_I(w)=0$, $p\ge u_I$, and $p(w_0^*)=u_I(w_0^*)=0$. Since $p$
is affine on $[0,v^*]$, it follows that $p(w)=0$ for all $w\in[0,v^*]$.

Hence on $[v^*,1]$ we must have
\[
p(w)=\beta (w-v^*)
\]
for some slope $\beta$. Since $w^*\in \operatorname{supp}G_{BHM}^*$, condition (iii) gives
\[
p(w^*)=u_I(w^*)=bq^*,
\]
so
\[
\beta=\frac{bq^*}{w^*-v^*}.
\]
But $p\ge u_I$ on $[w^*,1]$, and both functions are affine there, with $p(w^*)=u_I(w^*)$.
Therefore the slope of $p$ must be at least the slope of $u_I$, that is, $\beta\ge q^*$. This
implies
\[
\frac{bq^*}{w^*-v^*}\ge q^*
\quad\Longrightarrow\quad
b\ge w^*-v^*,
\]
a contradiction.

Thus $G_{BHM}^*$ solves the intermediary's problem if and only if
\[
b\ge b^*:=w^*-v^*.
\]
This proves the proposition.

\qed

\subsection*{Proof of Proposition 4}
Let $v_0^*:=0$ and $v_{N+1}^*:=1$. After eliminating any redundant items, we may assume
\[
q_0^*:=0,\qquad t_0^*:=0,\qquad q_i^*>q_{i-1}^* \ \text{ for all } i=1,\dots,N.
\]
Write
\[
\mu_i:=F_0(v_{i+1}^*)-F_0(v_i^*),\qquad
w_i^*:=E(\theta\mid v_i^*\le \theta\le v_{i+1}^*),\qquad i=0,\dots,N.
\]
Thus $G_{BHM}^*$ is the atomic distribution that places mass $\mu_i$ at $w_i^*$.

Suppose, toward a contradiction, that for the posted menu
\[
(q_1^*,t_1^*),\dots,(q_N^*,t_N^*)
\]
the distribution $G_{BHM}^*$ solves the intermediary's problem. The intermediary's indirect
utility is piecewise affine with finitely many jumps, so it is Lipschitz on a neighborhood of $0$
and on a neighborhood of $1$. Hence, by Deniz and Kovac (2020), there exists a price
function $p$ for $G_{BHM}^*$.

For each $i=0,\dots,N$, Jensen's inequality gives
\[
\int_{v_i^*}^{v_{i+1}^*} p(\theta)\,dF_0(\theta)
\ge
\mu_i\,p(w_i^*).
\]
Summing over $i$ and using condition (iv) in Definition 1, we obtain equality. Therefore
equality must hold interval-by-interval, which implies that $p$ is affine on every interval
$[v_i^*,v_{i+1}^*]$.

For $i=1,\dots,N$, define
\[
u_i(w):=(w+b)q_i^*-t_i^*,\qquad u_0(w):=0.
\]
Fix some $i\in\{1,\dots,N\}$ and let $\beta_i$ denote the slope of $p$ on the interval
$[v_i^*,v_{i+1}^*]$.

Because $w_i^*\in \operatorname{supp}G_{BHM}^*$, condition (iii) gives
\[
p(w_i^*)=u_i(w_i^*).
\]
Also, since $v_i^*<w_i^*$, a buyer with posterior mean $v_i^*$ chooses item $i-1$ (or the
outside option when $i=1$). Hence
\[
p(v_i^*)\ge u_{i-1}(v_i^*).
\]
It follows that
\[
\beta_i
=
\frac{p(w_i^*)-p(v_i^*)}{w_i^*-v_i^*}
\le
\frac{u_i(w_i^*)-u_{i-1}(v_i^*)}{w_i^*-v_i^*}.
\]
Using
\[
w_i^*(q_i^*-q_{i-1}^*)=t_i^*-t_{i-1}^*,
\]
we obtain
\[
u_i(w_i^*)-u_{i-1}(v_i^*)
=
b(q_i^*-q_{i-1}^*)+q_{i-1}^*(w_i^*-v_i^*).
\]
Therefore
\[
\beta_i
\le
q_{i-1}^*+\frac{b(q_i^*-q_{i-1}^*)}{w_i^*-v_i^*}.
\]

On the other hand, $p$ is affine on $[v_i^*,v_{i+1}^*]$, it satisfies
$p(w_i^*)=u_i(w_i^*)$, and it must lie above $u_i$ on $[w_i^*,v_{i+1}^*]$. Since $u_i$ is
affine on that interval with slope $q_i^*$, this is possible only if
\[
\beta_i\ge q_i^*.
\]
Combining the two inequalities yields
\[
q_i^*
\le
q_{i-1}^*+\frac{b(q_i^*-q_{i-1}^*)}{w_i^*-v_i^*}.
\]
Because $q_i^*>q_{i-1}^*$, we conclude that
\[
b\ge w_i^*-v_i^*.
\]
Since this must hold for every $i=1,\dots,N$, we get
\[
b\ge \max_{i=1,\dots,N}(w_i^*-v_i^*)=\hat b,
\]
contradicting the hypothesis $b<\hat b$.

Hence $G_{BHM}^*$ cannot solve the intermediary's problem. Therefore
$(q_{BHM}^*,G_{BHM}^*)$ violates (I-OB) and cannot be monopolist-optimal.
\qed

\subsection*{Proof of Proposition \ref{opt_types}}
\subsubsection*{Preliminaries}
In order to prove the Proposition, we will be using results established in \cite{mensch_ravid_2022}. First, for an increasing function $l:[0,1]\to [x,y]$, we say that $l $ is constant around $w$ if it is constant in some open neighbor of $w$. If it is not constant, we say $l$ is strictly increasing. We let $l_- (w)$ and $l_+ (w)$ be the left and right limits of $l$ at $w$, respectively. If $l $ is also differentiable, we let $l_+ (w)$ denote its right derivative at $w$.

Now, we say that an allocation $q$ jumps towards efficiency if

$$q(w)\in \argmax_{\tilde{q}\in [q_-(w), q_+ (w)]} [w\tilde{q} -c(\tilde{q}]$$

\cite{mensch_ravid_2022} show that one can take any allocation and replace it with an allocation that jumps towards efficiency without reducing the monopolist’s
profits. Consequently, focusing on allocations that jump towards efficiency is without loss
of optimality. Second, they show that allocations that jump towards efficiency convey a
technical benefit: the monopolist's profit function is an upper-semicontinuous function. Therefore, we restrict attention to allocations that jump towards efficiency. 

Next, we define a \textbf{bias-cancelling} mechanism. We call an increasing function $pd:[0,1]\to [0,\Bar{q}]$,  with the property that is constant around any $w$ at which $I_G (w)>0$ a $G$-marginal dual price function. We let $\underline{w}\equiv \min \supp G$ and $\Bar{w}\equiv \max\supp G$. The following Lemma is the direct analog of Lemma 1 in \cite{mensch_ravid_2022} for our environment, and its proof goes through, accounting for minor changes, verbatim.

\begin{lemma}\label{l_marginal_price}
Fix any allocation $q$. Then $G$ is a solution to the intermediary's problem if and only if a $G$-marginal price $pd$ exists such that the function

$$p_{q,pd}(w) := bq(\underline{w}) +\int^{\underline{w}}_0 q(\Bar{w}) d\Bar{w}+\int_{\underline{w}}^w pd(\Bar{w}) d\Bar{w}$$

lies weakly below $bq(w) +\int^{w}_0 q(\Bar{w}) d\Bar{w}$ for all $w$ and is equal to this expression for all $w\in \supp G$.
Moreover, one can choose $pd(w)=bq'_+ (w)+ q(w)$ for all $w\in\supp G$.
\end{lemma}
The key observation is that marginal dual price functions are (almost everywhere) derivatives of dual price functions. The requirement that marginal price
functions are increasing corresponds to the convexity of dual price functions. The restriction
that marginal prices are constant around signal realizations with a slack MPS constraint corresponds to price functions being affine over the same region. Consequently, one can
use $p_{q,pd}$ as a price function certifying the optimality of $G$ whenever $p_{q,pd}$ satisfies the
lemma’s desiderata. Conversely, whenever $G$ is optimal, one can obtain an $G$- marginal price function satisfying the lemma’s requirements by taking a derivative of the dual price function delivered by \cite{d_martini_2019}.

Now, define an allocation $q^{BC}$ as

\[
q^{BC}(w) = \begin{cases}
pd(w) - bq'_+ (w) & \text{if}~w \in (\underline{w}, \Bar{w}) \\
\max\{pd(\underline{w})-bq'_+ (w),0\} & \text{if}~ w\in [0,\underline{w}]\\
\min\{pd(\Bar{w})-bq'_+ (w),\Bar{q}\} & \text{if}~ w\in [\Bar{w},1] 
\end{cases}
\]

We say an allocation $q$ is $G$-bias-canceling if $q=q^{BC}$ for some $G$-marginal price $pd$. We refer to a mechanism as bias-cancelling if its allocation is $G$-bias-canceling. 
The following result is a direct analog of Theorem 2 in \cite{mensch_ravid_2022} in our environment. Again, its proof goes through verbatim accounting for minor details.

\begin{lemma}\label{BC}
It is without loss of optimality to restrict attention to bias-canceling mechanisms.

\end{lemma}

Equipped with these results, we are ready to prove the Proposition.

\begin{proof}
    By Lemma \ref{BC}, we focus on bias-cancelling mechanisms. Since we also restrict attention to finite item menus, it follows that the allocation $q$ is a step function. That is, we have

    \[
q(w) = \begin{cases}
0 & \text{if}~w \in [0, w_1) \\
q_1& \text{if}~ w\in [w_1,w_2)\\
\cdots \\
q_N & \text{if}~ w\in [w_N,1] 
\end{cases}
\]

    Thus, for all $w\in\supp G$ we have that $q'_+(w)=0$ and thus, by Lemma \ref{l_marginal_price} we can choose $pd(w)=q(w)$.
    This implies that the slope of the dual price function must be equal to the allocation for $w\in\supp G$, in order for $G$ to satisfy the (I-OB) constraint, given allocation $q$. Since the intermediary's indirect utility function $u^I(w)$ also has a slope equal to the allocation for all $w\in\supp G$, it follows that it must be the case that $u^I(w)=p(w)$ for all $w\in \bigcup_{i=1}^N [w_i,w_{i+1}]\cup [w_N,1]$. 

    Now, consider a finite item menu $(q_i,t_i)_{i=1}^N$ together with a distribution $G$. The previous argument, together with the remaining requirements for a valid dual price function, implies that the dual price function that makes $G$ satisfy the (I-OB) constraint must be given by 

    \[
p(w) = \begin{cases}
0 & \text{if}~w \in \left[0, \frac{t_1}{q_1}-b\right) \\
(w+b)q_1-t_1 & \text{if}~ w\in \left[\frac{t_1}{q_1}-b, \frac{t_2-t_1}{q_2-q_1}-b\right]\\
\cdots\\
(w+b)q_N-t_N & \text{if}~ w\in \left[\frac{t_N-t_{N-1}}{q_N-q_{N-1}}-b, 1\right]
\end{cases}
\]

and, thus, it must be the case that $G$ pools types in the intervals 
    \begin{align*}
    \left[0,\frac{t_1}{q_1}-b\right],\left[\frac{t_1}{q_1}-b, \frac{t_2-t_1}{q_2-q_1}-b\right],...,\left[\frac{t_i-t_{i-1}}{q_{i}-q_{i-1}}-b,\frac{t_{i+1}-t_{i}}{q_{i+1}-q_{i}}-b\right],...,\left[\frac{t_N-t_{N-1}}{q_{N}-q_{N-1}}-b,1\right]
    \end{align*}
    and has support given by
    
    \begin{equation*}
        \supp G= \{w_0,w_1,w_2,\cdots w_N\}
    \end{equation*}

where 
\begin{align*}
    w_0= & \EE\left(\theta|0\leq \theta\leq \frac{t_1}{q_1}-b\right)\\
    w_1= & \EE\left(\theta|\frac{t_1}{q_1}-b\leq\theta\leq \frac{t_2-t_1}{q_2-q_1}-b\right)\\
    w_i= & \EE\left(\theta|\frac{t_i-t_{i-1}}{q_{i}-q_{i-1}}-b\leq \theta\leq \frac{t_{i+1}-t_{i}}{q_{i+1}-q_{i}}-b\right) \ \ \text{for} \ \ i=2,...,N-1\\
    w_N= & \EE\left(\theta|\frac{t_N-t_{N-1}}{q_{N}-q_{N-1}}-b\leq \theta\leq 1\right)
    \end{align*}
and the expectation is taken with respect to the prior distribution $F_0$.

Thus, from now on, we need only consider finite-item menus and distributions that take the aforementioned form and look for the optimal among them. Therefore, for each $N$, for a monopolist who wants to implement an $N$ item menu, the problem reduces to one a la \cite{maskin_riley_1984} with the caveat that the probability of each type in the support depends on the menu the monopolist posts. 
 Let $v_1=t_1/q_1$, $v_i=(t_i-t_{i-1})/(q_i-q_{i-1})$ for $i=2,...,N$. Then, each posterior type happens with probabilities given by $\omega_0:=F_0(v_1-b)$, $\omega_i := F_0(v_i-b)-F_0(v_{i-1}-b)$, $i=2,\cdots, N$ and $\omega_N :=1-F_0( v_N-b)$.

 The monopolist's problem, thus, reduces to 

 \begin{align*}
    \max_{\{t_i\}_{i=1,\dots,N},\{q_i\}_{i=1,\dots,N}} \sum_{i=1}^{N-1}\left[ \omega_i(t_i-c(q_i)) \right]+\omega_N(t_N-c(q_N))
\end{align*}
subject to:
the (B-IR) constraints, $\{w_i q_i-t_i\}\geq 0$, $i=1,\cdots, N$, downward and upward (B-IC) constraints. Standard arguments imply that the only non-redundant constraints are the local downward (B-IC) constraints and the (B-IR) constraint for type $w_1$. 

We now argue that these constraints must be binding\footnote{We show this for the local downward (B-IC) constraints as the argument for the (B-IR) constraint for type $w_1$ is exactly the same}. Suppose that for some $i$, we have $w_i q_i-t_i>0$. Consider the following monopolist problem: the menu stays the same except $t_i$, and now the monopolist chooses $t_i$ to maximize the profit received from types in the interval $[v_i-b,v_{i+1}-b]$. The solution to this problem will be the same as the monopolist setting $t_i=t_{i-1}+\EE(\theta|v_i-b\leq \theta\leq v_{i+1})q_i$ and choosing $v_i$ to maximize profit from the interval $[v_i-b, v_{i+1}-b]$. Note that here we take the given value $v_{i+1}=(t_{i+1}-t_{i})/(q_{i+1}-q_i)$ from the original menu to calculate the expectation. The monopolist's problem is, thus, given by

\begin{align*}
    \max_{v_i} \left[ \left(F_0(v_{i+1}-b)-F_0(v_i-b) \right) \left(t_{i-1}+\EE(\theta|v_i-b\leq \theta\leq v_{i+1})q_i-c(q_i)\right) \right]
\end{align*}

Denote $\hat{v_i}$ the solution and $\hat{t}_i=t_{i-1}+\EE (\theta|\hat{v}_i-b\leq \theta\leq v_{i+1}-b)q_i$. Let $\epsilon:=\hat{t}_i-t_i$. Consider the menu $\{\Bar{t}_j,q_j\}$ where $\Bar{t}_j=t_j$ for $j<i$ and $\Bar{t}_j=t_j+\epsilon$ for $j\geq i$. By construction, the new menu satisfies all constraints, and the monopolist's profit under the new menu with binding downward IC constraint for type $i$ is higher than the original profit. Thus, local downward (B-IC) constraints must be binding at the optimum, which implies that $w_i=v_i$ for $i=1,\cdots,N$.

 Therefore, posterior types are obtained by starting from the last equation and obtaining $w_N$, then using this value to obtain $w_{N-1}$ and continuing this way until all $N$ posterior types are obtained. Notice that for each prior $F_0$ and each $N$, there is a unique distribution $G$ with support $\{w_0,w_1,\cdots,w_N\}$ that can be the posterior distribution of types induced by a finite item menu that satisfies the (I-OB) constraint. This follows by the increasing hazard rate property of $F_0$. Specifically, because $F_0$ has increasing hazard rate, the equation 
$w_N=  \EE\left(\theta|w_N-b\leq \theta\leq 1\right)$  has a unique solution and likewise, so does equation $w_{N-1}= \EE\left(\theta|w_{N-1}-b\leq \theta\leq w_N\right)$ that pins down $w_{N-1}$ and so on. Thus, in this sense, this constraint completely pins down the possible posterior buyer types.with the $N$ types of buyer corresponding to $\{w_0,w_1,\cdots,w_N\}$, with probabilities $F(w_1-b)$, $\left(F_0(w_{j+1}-b)-F_0(w_j-b)\right)$ for $j=1,\cdots,N-1$ and $(1-F_0 (w_n-b))$ respectively. It immediately follows that we must have the following: 

\begin{align*}
    w^*_0= & \EE(\theta|0\leq \theta\leq w_1^*-b)\\
    w^*_1 & = t^*_1/q^*_1=\EE(\theta|w_1^*-b\leq\theta\leq w_2^*-b)\\
    w^*_i & =(t^*_i-t^*_{i-1})/(q^*_i-q^*_{i-1})=\EE(\theta|w_i^*-b\leq \theta\leq w_{i+1}^*-b) \ \ \text{for} \ \ i=2,...,N-1\\
    w_N^* & = \frac{t_N^*-t^*_{N-1}}{q_N^*-q_{N-1}^*}=\EE(\theta|w_N^*-b\leq \theta\leq 1)
    \end{align*}

Since for any optimal $G$ the virtual valuations will be increasing, it immediately follows that the optimal qualities are given by    
 \begin{equation*}
    c'(q^*_N)=w_N^*  \tag{$\mathrm{OPT-Q_N}$}\\
\end{equation*}
\begin{equation*}
    c'(q^*_i)=w_i^*-\frac{\sum_{j=i}^{N-1} \left(F_0(w^*_{j+1}-b)-F_0(w^*_j-b)\right) }{F_0(w^*_{i+1}-b)-F_0(w^*_i-b)}(w^*_{i+1}-w^*_{i}) \ \ \text{for} \ \ i=1,\dots, N-1 \tag{$\mathrm{OPT-Q_i}$}
    \end{equation*}
The proof is now complete.
\end{proof}

\subsection*{Proof of Proposition 6}
For each $N\ge 1$ and each bias $b$, let $\Pi_N(b)$ denote the monopolist's maximal profit
among $N$-item menus, and let
\[
w_1^N(b),\dots,w_N^N(b)
\]
be the posterior types pinned down by Proposition 5, i.e.
\[
w_N^N(b)=E(\theta\mid w_N^N(b)-b\le \theta\le 1),
\]
and, for $i=1,\dots,N-1$,
\[
w_i^N(b)=E(\theta\mid w_i^N(b)-b\le \theta\le w_{i+1}^N(b)-b).
\]

A key observation is that
\[
w_i^N(b)=w_{i+1}^{N+1}(b)\qquad \text{for all } i=1,\dots,N.
\]
This follows immediately from the backward recursion above: the $(N+1)$-item problem adds
exactly one extra equation at the bottom, while the top $N$ equations coincide with those of
the $N$-item problem after an index shift.

Therefore the associated probabilities also satisfy
\[
\omega_i^N(b)=\omega_{i+1}^{N+1}(b)\qquad \text{for all } i=1,\dots,N.
\]

Now take any feasible $N$-item menu
\[
(q_1,\dots,q_N;\ t_1,\dots,t_N).
\]
Define an $(N+1)$-item menu by adding a dummy first item:
\[
\tilde q_1=0,\qquad \tilde t_1=0,\qquad
\tilde q_{i+1}=q_i,\qquad \tilde t_{i+1}=t_i,\qquad i=1,\dots,N.
\]
Because the type/probability vectors are nested as above, this $(N+1)$-item menu is feasible
and yields exactly the same profit as the original $N$-item menu. Hence
\[
\Pi_{N+1}(b)\ge \Pi_N(b)\qquad \text{for every } N.
\]

If the optimal $(N+1)$-item menu has $q_1^{N+1}(b)=0$, then by IR also
$t_1^{N+1}(b)=0$, and removing the first item leaves an $N$-item menu with exactly the same
profit. Hence in that case
\[
\Pi_{N+1}(b)=\Pi_N(b).
\]

If instead $q_1^{N+1}(b)>0$, then the optimal $(N+1)$-item menu strictly improves on the
dummy-item embedding of the $N$-item optimum, so
\[
\Pi_{N+1}(b)>\Pi_N(b).
\]

Therefore the seller's optimal finite menu contains exactly the largest number of
non-degenerate items, i.e. the largest $N$ such that the lowest allocated quality is strictly
positive.

Now, the fact that $N^*_b$ is a decreasing function of $b$ follows simply by observing that as $b$ decreases, all types specified by Proposition 5 increase. 

\begin{lemma}
Fix $N$. For every $i=1,\dots,N$, the posterior type $w_i^N(b)$ is strictly decreasing in $b$.
\end{lemma}

\begin{proof}
Set $x_i^N(b):=w_i^N(b)-b$. Then $x_N^N(b)$ solves
\[
E(\theta\mid x_N^N(b)\le \theta\le 1)-x_N^N(b)=b,
\]
and, for $i=1,\dots,N-1$, $x_i^N(b)$ solves
\[
E(\theta\mid x_i^N(b)\le \theta\le x_{i+1}^N(b))-x_i^N(b)=b.
\]
Define
\[
R(x,y):=E(\theta\mid x\le \theta\le y)-x.
\]
Because $f_0>0$ on $[0,1]$, the function $R(x,y)$ is strictly decreasing in $x$ and strictly
increasing in $y$. A backward induction on $i$ therefore gives
\[
x_i^N(b')>x_i^N(b)\qquad \text{whenever } b'<b.
\]
Since
\[
w_i^N(b)=E(\theta\mid x_i^N(b)\le \theta\le x_{i+1}^N(b))
\]
for $i<N$ and
\[
w_N^N(b)=E(\theta\mid x_N^N(b)\le \theta\le 1),
\]
it follows that $w_i^N(b')>w_i^N(b)$ whenever $b'<b$.
\end{proof}

Thus, the lowest positive quality that solves (OPT-Q\textsubscript{1}) is attained for a higher number of items. \qed

\end{document}